\documentclass[english,twocolumn,superscriptaddress,letterpaper,nofootinbib]{revtex4}
\usepackage{ae} % or {zefonts}0
\usepackage[T1]{fontenc}
\usepackage[ansinew]{inputenc}
\usepackage[english]{babel}
\usepackage{amsmath,amssymb}
\usepackage[all,cmtip]{xy}
\usepackage[colorlinks]{hyperref}
\usepackage{graphicx}
\usepackage{color}
\usepackage{paralist}

\usepackage{mathrsfs}
\DeclareMathAlphabet{\mathsc}{U}{rsfs}{m}{n}

\DeclareMathOperator{\Ann}{Ann}
\DeclareMathOperator{\Der}{Der}
\DeclareMathOperator{\Mat}{Mat}
\DeclareMathOperator{\Sl}{Sl}

\DeclareMathOperator{\Skw}{Skw}
\DeclareMathOperator{\Sym}{Sym}
\DeclareMathOperator{\syz}{syz}
\DeclareMathOperator{\tr}{tr}

\newcommand{\CC}{\mathbb{C}}
\newcommand{\QQ}{\mathbb{Q}}
\renewcommand{\d}{\mathrm{d}}
\newcommand{\J}{\mathsc{J}}
\newcommand{\KK}{\mathbb{K}}
\newcommand{\NN}{\mathbb{N}}
\renewcommand{\O}{\mathsc{O}}
\newcommand{\p}{\partial}
\newcommand{\xmid}{\;\middle|\;}

\newcommand{\ideal}[1]{{\left\langle#1\right\rangle}}
\newcommand{\set}[1]{{\left\{#1\right\}}}

\usepackage{amsthm}
\theoremstyle{plain}
\newtheorem{thm}{Theorem}%[section]
\newtheorem{cor}[thm]{Corollary}
\newtheorem{prp}[thm]{Proposition}

\makeatletter
\newcommand{\pushright}[1]{\ifmeasuring@#1\else\omit\hfill$\displaystyle#1$\fi\ignorespaces}
\newcommand{\pushleft}[1]{\ifmeasuring@#1\else\omit$\displaystyle#1$\hfill\fi\ignorespaces}

\numberwithin{equation}{section}

\begin{document}
%%%%%%%%%%%%%%%%%%%%%%%%%%%%%%%%%%%%%%%%%%%%%%%%%%%%%%%%%%%%%%%%%%%%%%%%%%%%%%%%
\phantom{.}
\newpage
\hfill\hbox{MITP/17-104, UUITP-44/17}
\vspace{5mm}

\title{Complete sets of logarithmic vector fields \\ for integration-by-parts identities of Feynman integrals}

\author{Janko B{\"o}hm}\email{boehm@mathematik.uni-kl.de}
\affiliation{Department of Mathematics, TU Kaiserslautern, 67663 Kaiserslautern, Germany}
\author{Alessandro Georgoudis}\email{Alessandro.Georgoudis@physics.uu.se}
\affiliation{Department of Physics and Astronomy, Uppsala University, SE-75108 Uppsala, Sweden}
\author{Kasper J. Larsen}\email{Kasper.Larsen@soton.ac.uk}
\affiliation{School of Physics and Astronomy, University of Southampton, Highfield, Southampton, \\ SO17 1BJ, United Kingdom}
\author{Mathias Schulze}\email{mschulze@mathematik.uni-kl.de}
\affiliation{Department of Mathematics, TU Kaiserslautern, 67663 Kaiserslautern, Germany}
\author{Yang Zhang}\email{zhang@uni-mainz.de}
\affiliation{ETH Z{\"u}rich, Wolfang-Pauli-Strasse 27, 8093 Z{\"u}rich, Switzerland}
\affiliation{PRISMA Cluster of Excellence, Johannes Gutenberg University, 55128 Mainz, Germany}

\begin{abstract}
Integration-by-parts identities between loop integrals arise
from the vanishing integration of total derivatives in dimensional
regularization. Generic choices of total derivatives in the Baikov
or parametric representations lead to identities which involve
dimension shifts. These dimension shifts can be avoided by imposing
a certain constraint on the total derivatives. The solutions of
this constraint turn out to be a specific type of syzygies
which correspond to logarithmic vector fields along the Gram determinant
formed of the independent external and loop momenta. We present
an explicit generating set of solutions in Baikov representation,
valid for any number of loops and external momenta, obtained
from the Laplace expansion of the Gram determinant.
We provide a rigorous mathematical proof that this set of
solutions is complete. This proof relates the logarithmic vector fields
in question to ideals of submaximal minors of the Gram matrix
and makes use of classical resolutions of such ideals.
\end{abstract}

\maketitle

%%%%%%%%%%%%%%%%%%%%%%%%%%%%%%%%%%%%%%%%%%%%%%%%%%%%%%%%%%%%%%%%%%%%%%%%%%%%%%%%
\section{Introduction}
%%%%%%%%%%%%%%%%%%%%%%%%%%%%%%%%%%%%%%%%%%%%%%%%%%%%%%%%%%%%%%%%%%%%%%%%%%%%%%%%

The increasing precision of the experimental measurements
of scattering processes at the Large Hadron Collider (LHC)
is calling for increased precision in the
theoretical prediction of cross sections.
The computations of the leading-order (LO) and next-to-leading-order (NLO)
contributions to the cross sections are by now automated,
but for many processes the next-to-next-to-leading-order (NNLO)
contribution is needed to reach the required precision.

The NNLO cross section has double-real, real-virtual and
double-virtual contributions. The aim of this paper is to
provide new tools for computing the latter contributions,
i.e., the two-loop scattering amplitudes. Results for the latter
may in turn motivate progress on the combination of all
virtual and real contributions to the NNLO cross section.

A key tool in the calculation of multi-loop amplitudes are
integration-by-parts (IBP) reductions. The latter arise
from the vanishing integration of total derivatives
in dimensional regularization,
\begin{equation}
\int \prod_{i=1}^L \frac{\d^D \ell_i}{\mathrm{i} \pi^{D/2}}
\sum_{j=1}^L \frac{\partial}{\partial \ell_j^\mu}
\frac{v_j^\mu \hspace{0.5mm} P}{D_1^{\nu_1} \cdots D_m^{\nu_m}}
\hspace{1mm}=\hspace{1mm} 0 \,, \label{eq:IBP_schematic}
\end{equation}
where $P$ and the vectors $v_j^\mu$ are polynomial in the loop momenta $\ell_i$ and
external momenta, the $D_k$ denote inverse propagators, and the
$\nu_i$ are integers. The IBP identities turn out to
generate a large set of linear relations between loop integrals.
This then allows for most integrals to be reexpressed as a linear
combination of basis integrals. In practice, the basis contains much fewer integrals
than the number of integrals produced by the Feynman rules
for a multi-loop amplitude.

The step of performing Gaussian elimination on the linear systems
that arise from eq.~\eqref{eq:IBP_schematic} may be carried out
with the Laporta algorithm~\cite{Laporta:2000dc,Laporta:2001dd},
which leads in general to relations that involve integrals with
squared propagators. There are various publicly available implementations
of automated IBP reduction: AIR~\cite{Anastasiou:2004vj},
FIRE~\cite{Smirnov:2008iw,Smirnov:2014hma},
Reduze~\cite{Studerus:2009ye,vonManteuffel:2012np},
LiteRed~\cite{Lee:2012cn}, Kira~\cite{Maierhoefer:2017hyi},
as well as private implementations.
A formalism for obtaining IBP reductions without squared
propagators was developed in refs.~\cite{Gluza:2010ws,Ita:2015tya}.
A systematic method of deriving IBP reductions on
generalized-unitarity cuts was given in ref.~\cite{Larsen:2015ped}.
A recent approach~\cite{vonManteuffel:2014ixa} uses sampling
over finite fields to construct the reduction coefficients.
Other recent developments include software for determining a
basis of integrals \cite{Georgoudis:2016wff} and a
D-module theory based approach for computing the number
of basis integrals \cite{Bitoun:2017nre}.

The IBP reductions moreover allow setting up differential equations
for the basis integrals, thereby enabling their evaluation.%
~\cite{Kotikov:1990kg,Kotikov:1991pm,Bern:1993kr,Remiddi:1997ny,Gehrmann:1999as,Henn:2013pwa,%
Papadopoulos:2014lla,Ablinger:2015tua,Liu:2017jxz}.
Differential equations have proven to be a powerful tool
for calculating multi-loop integrals, enabling for example
the computation of the basis integrals for numerous
two-loop amplitudes of $2 \to 2$ processes. This method
can therefore reasonably be expected to be of relevance
to two-loop amplitudes for $2 \to 3$ processes. We note that,
in the context of the latter, impressive results have recently
appeared \cite{Gehrmann:2015bfy,Papadopoulos:2015jft,Badger:2017jhb,Abreu:2017hqn}.

In this paper we study integration-by-parts identities \eqref{eq:IBP_schematic}.
Generic choices of total derivatives in the Baikov
or parametric representations lead to identities which involve undesirable
dimension shifts. These dimension shifts can be avoided by imposing
a certain constraint on the total derivatives. The solutions of
this constraint turn out to be a specific type of syzygies
which correspond to logarithmic vector fields along the Gram determinant
formed of the independent external and loop momenta. We will present
an explicit generating set of solutions in Baikov representation,
valid for any number of loops and external momenta, obtained
from the Laplace expansion of the Gram determinant.
We will then present a rigorous mathematical proof that this set of
solutions is complete. This proof relates the logarithmic vector fields
in question to ideals of submaximal minors of the Gram matrix
and makes use of classical resolutions of such ideals.

An important feature of the obtained generating set of syzygies
is that they are guaranteed to have degree at most one. In contrast, a generating
set of syzygies obtained from an S-polynomial computation would
in general have higher degrees. The fact that the syzygies obtained
here are of degree at most one is useful because it dramatically
simplifies the computation of solutions which satisfy
further constraints. For example, one may be interested in
imposing the further constraint on the total derivatives that
no integrals with squared propagators
are encountered in the integration-by-parts identities.

This paper is organized as follows.
In Sec.~\ref{sec:Baikov_representation}
we set up notation and give the Baikov representation
of a generic Feynman loop integral. In Sec.~\ref{sec:IBP_identities_on_cuts}
we study integration-by-parts relations on unitarity cuts
and derive the syzygy equation of interest.
In Sec.~\ref{sec:syzygies_from_Laplace_expansion} we obtain
a closed-form generating set of solutions to the syzygy equation.
In Sec.~\ref{sec:proof_of_completeness_of_syzygies}
we present a proof that the set of syzygies is complete.
In Sec.~\ref{sec:examples} we provide an example
of the formalism.
In Sec.~\ref{sec:conclusions} we give our conclusions.

%%%%%%%%%%%%%%%%%%%%%%%%%%%%%%%%%%%%%%%%%%%%%%%%%%%%%%%%%%%%%%%%%%%%%%%%%%%%%%%%
\section{Baikov representation of loop integrals}\label{sec:Baikov_representation}
%%%%%%%%%%%%%%%%%%%%%%%%%%%%%%%%%%%%%%%%%%%%%%%%%%%%%%%%%%%%%%%%%%%%%%%%%%%%%%%%

In this paper we will make use of the Baikov representation
of Feynman loop integrals. As will be explained later, this
parametrization is useful for our purpose of studying
integration-by-parts identities \eqref{eq:IBP_schematic}
on so-called cuts where some number of propagators are
put on shell, i.e. after evaluating the residue at $D_\alpha = 0$.
Since the Baikov representation uses the inverse propagators
$D_\alpha$ as variables, it greatly facilitates the application of cuts.

In this section we fix our notations
and review the Baikov representation of a general Feynman loop integral.
We consider an integral with $L$ loops and $k$ propagators.
We denote the loop momenta as $\ell_1, \ldots, \ell_L$
and the external momenta as $p_1, \ldots, p_E, p_{E+1}$,
where $E$ thus denotes the number of linearly independent external
momenta. Furthermore, the integrand may involve $m-k$
irreducible scalar products---that is, polynomials in the
loop momenta and external momenta which cannot be expressed
as a linear combination of the inverse propagators. As will
be shown below, $m$ is a function of $L$ and $E$.
We apply dimensional regularization to regulate infrared and ultraviolet
divergences and normalize the integral as follows,
\begin{equation}
I(\nu_1, \ldots, \nu_m; D) \equiv \int \prod_{j=1}^L \frac{\d^D \ell_j}{\mathrm{i} \pi^{D/2}}
\frac{N_{k,m}}{D_1^{\nu_1} \cdots D_k^{\nu_k}} \,, \hspace{4mm} \nu_i \geq 0 \,.
\label{eq:def_generic_Feynman_integral}
\end{equation}
The inverse propagators $D_j$ are of the form $P^2$ where
$P$ is an integer-coefficient linear combination of vectors taken from
the ordered set of all independent external and loop momenta,
\begin{equation}
V = ( v_1, \ldots, v_{E+L}) = ( p_1, \ldots, p_E, \ell_1, \ldots, \ell_L ) \,.
\label{eq:list_of_momenta}
\end{equation}
Furthermore, the quantity $N_{k,m}$
in eq.~\eqref{eq:def_generic_Feynman_integral} is defined as
$N_{k,m} \equiv D_{k+1}^{\nu_{k+1}} \cdots D_m^{\nu_m}$.

We now proceed to present the Baikov representation \cite{Baikov:1996rk} of
the integral \eqref{eq:def_generic_Feynman_integral}. To this end, we start
by writing down the Gram matrix $S$ of the independent external and loop momenta,
\begin{equation}
S =
\left(
 \begin{array}{ccc|ccc}
  x_{1,1}  &  \hspace{-1mm} \cdots \hspace{-1mm}  &  \hspace{-1mm} x_{1,E}    &  x_{1,E+1}    &  \hspace{-1mm} \cdots \hspace{-1mm}  &  \hspace{-1mm} x_{1,E+L} \\
  \vdots   &  \hspace{-1mm} \ddots \hspace{-1mm}  &  \hspace{-1mm} \vdots     &  \vdots       &  \hspace{-1mm} \ddots \hspace{-1mm}  &  \hspace{-1mm} \vdots  \\
  x_{E,1}  &  \hspace{-1mm} \cdots \hspace{-1mm}  &  \hspace{-1mm} x_{E,E}    &  x_{E,E+1}    &  \hspace{-1mm} \cdots \hspace{-1mm}  &  \hspace{-1mm} x_{E,E+L} \\
\hline
 x_{E+1,1} &  \hspace{-1mm} \cdots \hspace{-1mm}  &  \hspace{-1mm} x_{E+1,E}  &  x_{E+1,E+1}  &  \hspace{-1mm} \cdots \hspace{-1mm}  &  \hspace{-1mm} x_{E+1,E+L} \\
 \vdots    &  \hspace{-1mm} \ddots \hspace{-1mm}  &  \hspace{-1mm} \vdots     &  \vdots       &  \hspace{-1mm} \ddots \hspace{-1mm}  &  \hspace{-1mm} \vdots  \\
 x_{E+L,1} &  \hspace{-1mm} \cdots \hspace{-1mm}  &  \hspace{-1mm} x_{E+L,E}  &  x_{E+L,E+1}  &  \hspace{-1mm} \cdots \hspace{-1mm}  &  \hspace{-1mm} x_{E+L,E+L}
\end{array} \right) \,,
\label{eq:extended_Gram_matrix}
\end{equation}
where the entries are given by,
\begin{equation}
x_{i,j} = v_i \cdot v_j \,, \label{eq:x_ij_definition}
\end{equation}
where $v_i$ and $v_j$ are entries of $V$ in eq.~\eqref{eq:list_of_momenta}.
In addition, we let $F$ denote the determinant of $S$,
\begin{equation}
F \equiv \det S \,.
\label{eq:definition_of_Baikov_polynomial}
\end{equation}

The entries of the upper-left $E \times E$ block of $S$
are constructed out of the external momenta only, and
it will be convenient for the following to emphasize
this by relabeling these entries,
\begin{equation}
\lambda_{i,j}  =  x_{i,j} \hspace{5mm} \mathrm{for} \hspace{5mm} 1 \leq i, j \leq E \,.
\label{eq:definition_of_lambda}
\end{equation}
Furthermore we define $G$ as the Gram matrix of
the independent external momenta,
\begin{equation}
G =
\begin{pmatrix}
\lambda_{1,1} & \cdots & \lambda_{1,E} \\
\vdots & \ddots & \vdots \\
\lambda_{E,1} & \cdots & \lambda_{E,E}
\end{pmatrix} \,,
\label{eq:Gram_matrix_of_external_momenta}
\end{equation}
\clearpage
\noindent and let $U$ denote its determinant,
\begin{equation}
U = \det G \,. \label{eq:definition_of_U}
\end{equation}
We remark that $U$ is equal to the square of the volume of the
parallelotope formed by the independent external momenta $\{ p_1, \ldots, p_E\}$.
Thus, $U$ is non-vanishing provided that $p_1, \ldots, p_E$
are not linearly dependent.

The entries of the remaining blocks of $S$ depend on the loop momenta.
As $S$ is a symmetric matrix, not all entries are independent.
We can choose as a set of independent entries
for example the entries of the upper-right $E \times L$ block
along with the upper-triangular entries of the lower-right $L \times L$ block,
\begin{equation}
x_{i,j} \hspace{3mm} \mathrm{where} \hspace{2mm}
\left\{  \begin{array}{l}
1 \leq i \leq E \hspace{3mm} \mathrm{and} \hspace{3mm} E{+}1 \leq j \leq E{+}L \,, \\[2mm]
E{+}1 \leq i \leq j \leq E{+}L \,.
\end{array}  \right.
\label{eq:independent_loop_containing_entries}
\end{equation}
Hence we find that $S$ contains $LE$ + $\frac{L(L+1)}{2}$ independent
entries which depend on the loop momenta. From the fact that each
inverse propagator $D_\alpha$ is the square of a linear combination
of the elements of $V$ in eq.~\eqref{eq:list_of_momenta} and the fact
that the elements of $V$ are linearly independent, it follows that
$D_\alpha$ can be written as a unique linear combination of the $x_{i,j}$
in eq.~\eqref{eq:independent_loop_containing_entries}.
We therefore conclude that the combined number of propagators and
irreducible scalar products in eq.~\eqref{eq:def_generic_Feynman_integral}
is given by the expression,
\begin{equation}
m = LE + \frac{L(L+1)}{2} \,.
\label{eq:number_of_independent_loop_momentum_dot_products}
\end{equation}
Keeping the relabeling in eq.~\eqref{eq:definition_of_lambda}
in mind, we can write any inverse propagator $D_\alpha$
(with $\alpha = 1,\ldots, m$) as an explicit linear
combination of the $x_{i,j}$
in eq.~\eqref{eq:independent_loop_containing_entries}
as follows,
\begin{equation}
D_\alpha = \sum_{\beta=1}^m A_{\alpha,\beta} x_\beta + \sum_{1\leq k \leq l \leq E} (B_\alpha)_{k,l} \lambda_{k,l}
\hspace{1mm}-\hspace{0.5mm} m_\alpha^2 \,, \label{eq:relation_of_z_to_x}
\end{equation}
where $A_{\alpha, \beta}$ and the entries of $B_\alpha$ are
integers. In writing this expression we introduced a
lexicographic order on the set of elements $(i,j)$ in
eq.~\eqref{eq:independent_loop_containing_entries} and
let $\beta = 1, \ldots, m$ denote the element label
in the ordered set.

The variables of the Baikov representation \cite{Baikov:1996rk}
are chosen as the inverse propagators and the irreducible scalar products,
\begin{equation}
z_\alpha \equiv D_\alpha \hspace{4mm} \mathrm{where} \hspace{4mm} 1 \leq \alpha \leq m \,.
\label{eq:definition_of_z}
\end{equation}
We can now present the Baikov representation of the integral in eq.~\eqref{eq:def_generic_Feynman_integral}.
It takes the following form%
\footnote{
We remark that the Baikov representation in eq.~\eqref{eq:Baikov_representation}
is consistent with that used in ref.~\cite{Larsen:2015ped}.
This is a consequence of the identity $\det_{i,j=1,\ldots,L} \mu_{i,j} = \frac{F}{U}$,
which in turn follows from the Schur complement theorem in linear algebra.
Moreover, ref.~\cite{Larsen:2015ped} makes use of the four-dimensional helicity
scheme. It is therefore imposed as a constraint on
the external momenta that they span a vector space of
dimension at most four. In order words, one must have
$\mathrm{dim} \hspace{0.9mm} \mathrm{span} \{ p_1, \ldots, p_E\} \leq 4$.
Accordingly, the exponent of the Baikov polynomial $F$
is modified from $\frac{D-L-E-1}{2}$ in eq.~\eqref{eq:Baikov_representation}
to $\frac{D-L-5}{2}$ and $\frac{D-L-4}{2}$
for $E \geq 4$ and $E=3$, respectively.
},
\begin{align}
I(\nu; D)  \hspace{0.5mm}&=\hspace{0.5mm}  C_E^L \hspace{0.7mm} U^\frac{E-D+1}{2} \hspace{-1mm}
\int \frac{\d z_1 \cdots \d z_m}{z_1^{\nu_1} \cdots z_k^{\nu_k}} F^\frac{D-L-E-1}{2} N_{k,m} \,,
\label{eq:Baikov_representation}
\end{align}
where the first prefactor is given by the expression,
\begin{equation}
C_E^L \hspace{0.7mm} \equiv \hspace{0.7mm} \frac{\pi^{-L(L-1)/4 - L E/2}}
{\prod_{j=1}^L \Gamma \hspace{-0.5mm} \left(\frac{D-L-E+j}{2} \right)} \det A \,,
\label{eq:Baikov_prefactor}
\end{equation}
where $A$ is the matrix defined in eq.~\eqref{eq:relation_of_z_to_x}.

%%%%%%%%%%%%%%%%%%%%%%%%%%%%%%%%%%%%%%%%%%%%%%%%%%%%%%%%%%%%%%%%%%%%%%%%%%%%%%%%
\section{Integration-by-parts identities on unitarity cuts}\label{sec:IBP_identities_on_cuts}
%%%%%%%%%%%%%%%%%%%%%%%%%%%%%%%%%%%%%%%%%%%%%%%%%%%%%%%%%%%%%%%%%%%%%%%%%%%%%%%%

In this section we consider integration-by-parts identities
\eqref{eq:IBP_schematic} on cuts where some number of propagators
are put on shell, i.e. roughly speaking $\frac{1}{D_i} \to \delta(D_i)$.
This has the advantage of reducing the linear systems to which
Gauss-Jordan elimination is to be applied. As explained in ref.~\cite{Larsen:2015ped},
it is possible to determine complete integration-by-parts reductions
by performing the reductions on a suitably chosen spanning set of cuts
and merge the information found on each cut.

The virtue of the Baikov representation \eqref{eq:Baikov_representation}
is that it makes manifest the effect of cutting propagators.
Cf.~refs.~\cite{Ita:2015tya,Larsen:2015ped}, we consider
applying a $c$-fold cut (where $0\leq c \leq k$)
to eq.~\eqref{eq:Baikov_representation}. We let $\mathcal{S}_\mathrm{cut}$,
$\mathcal{S}_\mathrm{uncut}$ and $\mathcal{S}_\mathrm{ISP}$ denote
the sets of indices labeling cut propagators, uncut propagators and
irreducible scalar products respectively, and set,
\begin{align}
\mathcal{S}_\mathrm{cut}    &=  \{ \zeta_1, \ldots, \zeta_c \} \nonumber \\
\mathcal{S}_\mathrm{uncut}  &=  \{ r_1, \ldots, r_{k-c}\} \\
\mathcal{S}_\mathrm{ISP}    &=  \{ r_{k-c+1}, \ldots, r_{m-c}\} \nonumber \,.
\end{align}
We will restrict the analysis to the case where the
propagator powers in eq.~\eqref{eq:def_generic_Feynman_integral}
are equal to one, $\nu_i = 1$.

The result of applying the cut,
\begin{equation}
\int \frac{\d z_i}{z_i} \hspace{4mm} \mathop{\longrightarrow^{\hspace{-5.7mm} \mathrm{cut}}} \hspace{4mm}
\oint_{\Gamma_\epsilon (0)} \frac{\d z_i}{z_i}
\hspace{5mm} \mathrm{where} \hspace{5mm} i \in \mathcal{S}_\mathrm{cut} \,,
\end{equation}
where $\Gamma_\epsilon (0)$ denotes a circle centered at $0$ of radius $\epsilon > 0$
to eq.~\eqref{eq:Baikov_representation} is obtained by evaluating the residue at
$z_i = 0$ where $i \in \mathcal{S}_\mathrm{cut}$,
\begin{align}
I_\mathrm{cut} (\nu; D) \hspace{0.5mm}&=\hspace{0.5mm} C_E^L(D) U^\frac{E-D+1}{2} \nonumber \\
& \hspace{-10mm} \times \int \frac{\d z_{r_1} \cdots \d z_{r_{m-c}} N_{k,m}}{z_{r_1} \cdots z_{r_{k-c}}}
F(z)^\frac{D-L-E-1}{2} \Big|_{z_i = 0, \hspace{1mm} i \in \mathcal{S}_\mathrm{cut}}\,.
\label{eq:Baikov_representation_on_cut}
\end{align}

We now turn to integration-by-parts identities evaluated
on the $c$-fold cut $\mathcal{S}_\mathrm{cut}$. Such identities
correspond to exact differential forms of degree $m-c$. The most
general exact differential form which is of the form
of the integrand of eq.~\eqref{eq:Baikov_representation_on_cut} is,
\begin{align}
0 &= \int \d \Big( \sum_{i=1}^{m-c} \frac{(-1)^{i+1} a_{r_i} F(z)^\frac{D-L-E-1}{2}}
{z_{r_1} \cdots z_{r_{k-c}}} \nonumber \\
& \hspace{14mm} \times \d z_{r_1} \wedge \cdots \wedge
\widehat{\d z_{r_i}} \wedge \cdots \wedge \d z_{r_{m-c}} \Big) \,,
\label{eq:IBP_relation_differential_form}
\end{align}
where the $a_i$ are polynomials in $\{ z_{r_1}, \ldots, z_{r_{m-c}} \}$.
Expanding eq.~\eqref{eq:IBP_relation_differential_form},
we get an integration-by-parts identity,
\begin{align}
0 &= \int \hspace{-0.5mm} \Big[ \sum_{i=1}^{m-c} \hspace{-0.8mm} \Big( \frac{\partial a_{r_i}}{\partial z_{r_i}}
+ \frac{D{-}L{-}E{-}1}{2 F(z)} a_{r_i} \frac{\partial F}{\partial z_{r_i}} \Big)
- \sum_{i=1}^{k-c} \frac{a_{r_i}}{z_{r_i}} \Big] \nonumber \\
&\hspace{12mm}\times \frac{F(z)^\frac{D{-}L{-}E{-}1}{2}}{z_{r_1} \cdots z_{r_{k-c}}}
\d z_{r_1} \hspace{-1mm} \cdots \d z_{r_{m-c}} \,.
\label{eq:IBP_relation_ansatz}
\end{align}
We observe that, for an arbitrary choice of polynomials $a_i(z)$,
the two terms in the parenthesis $(\cdots)$
in eq.~\eqref{eq:IBP_relation_ansatz} correspond to
integrals in $D$ and $D-2$ dimensions, respectively.
This is because the $\frac{1}{F(z)}$ factor in the
second term has the effect of modifying the integration measure,
thereby shifting the space-time dimension from $D$ to $D-2$.

To get the exact form in eq.~\eqref{eq:IBP_relation_differential_form}
to correspond to an integration-by-parts relation in $D$ dimensions,
we require the $a_i(z)$ to be chosen such that,
\begin{equation}
b F + \sum_{i=1}^{m-c} a_{r_i} \frac{\partial F}{\partial z_{r_i}} = 0 \,,
\label{eq:syzygy_equation}
\end{equation}
where $b$ denotes a polynomial, since then the $\frac{1}{F}$ factor
in eq.~\eqref{eq:IBP_relation_ansatz} cancels out, and no dimension shift occurs.
Equations of the type \eqref{eq:syzygy_equation} are known in algebraic geometry as
\emph{syzygy equations} (describing in our setting the polynomial relations---that is,
syzygies, between $F, \frac{\partial F}{\partial z_{r_1}},\ldots,\frac{\partial F}{\partial z_{r_{m-c}}}$). They have also been considered
in the context of integration-by-parts relations in
refs.~\cite{Gluza:2010ws,Schabinger:2011dz,Lee:2014tja,Ita:2015tya,Larsen:2015ped,Bern:2017gdk}.
We remark that it follows from Schreyer's theorem
that a generating set of solutions of eq.~\eqref{eq:syzygy_equation}
can be found algebraically by determining a Gr{\"o}bner basis of the ideal generated by the above polynomials, considering the S-polynomials involved in the Buchberger test, and expressing the corresponding relations in terms of the original generators \cite{opac-b1094391}.
We refer to refs.~\cite{Ita:2015tya,Georgoudis:2016wff} for a geometric
interpretation of eq.~\eqref{eq:syzygy_equation}.

%%%%%%%%%%%%%%%%%%%%%%%%%%%%%%%%%%%%%%%%%%%%%%%%%%%%%%%%%%%%%%%%%%%%%%%%%%%%%%%%
\section{Syzygy generators from Laplace expansion}\label{sec:syzygies_from_Laplace_expansion}
%%%%%%%%%%%%%%%%%%%%%%%%%%%%%%%%%%%%%%%%%%%%%%%%%%%%%%%%%%%%%%%%%%%%%%%%%%%%%%%%

In this section we turn to obtaining a generating set
$\mathcal{T} = \langle g_1, \ldots, g_d\rangle$
of syzygies $g_i = (a_{r_1}, \ldots, a_{r_{m-c}}, b)$ of
eq.~\eqref{eq:syzygy_equation}. By this we mean that
$\mathcal{T}$ is such that any solution of eq.~\eqref{eq:syzygy_equation}
can be written in the form $g_i p$ where $g_i \in \mathcal{T}$
and $p$ denotes a polynomial.

For a general polynomial $F$, determining a generating
set of syzygies would require an $S$-polynomial computation.
However, as we will shortly see, in the case where $F$ is
the determinant of a matrix, a generating set of
syzygies can be obtained from the Laplace expansion
of the determinant of $F$. We remark that related work
has appeared in ref.~\cite{Bern:2017gdk}.

%%%%%%%%%%%%%%%%%%%%%%%%%%%%%%%%%%%%%%%%%%%%%%%%%%%%%%%%%%%%%%%%%%%%%%%%%%%%%%%%
\subsection{Off-shell case}
%%%%%%%%%%%%%%%%%%%%%%%%%%%%%%%%%%%%%%%%%%%%%%%%%%%%%%%%%%%%%%%%%%%%%%%%%%%%%%%%

For simplicity we start with the case where no cuts are applied, $c=0$.
Let $M = (m_{i,j})_{i,j=1,\ldots, n}$ be a generic matrix, i.e. such
that all entries are independent. We consider the determinant of $M$
and perform Laplace expansion of the determinant along the $i$th row,
\begin{equation}
\left[ \sum_{k=1}^n m_{j,k} \frac{\partial (\det M)}{\partial m_{i,k}} \right] - \delta_{i,j} \det M = 0 \,,
\hspace{4mm} 1 \leq i,j \leq n \,.
\label{eq:Laplace_expansion_of_generic_matrix}
\end{equation}
The identities with $i\neq j$ follow by replacing the $i$th row of $M$
by the $j$th row, $m_{i,k} \to m_{j,k}$, as the resulting
matrix clearly has a vanishing determinant.

For a symmetric matrix $S = (s_{i,j})_{i,j=1,\ldots, n}$, the entries satisfy
$s_{i,j}=s_{j,i}$ and are thus not all independent.
For this case, one obtains from the Laplace expansion the following identities,
\begin{equation}
\left[\sum_{k=1}^n (1{+}\delta_{i,k}) s_{j,k} \frac{\partial (\det S)}{\partial s_{i,k}}\right] - 2\delta_{i,j} \det S = 0 \,,
\hspace{2mm} 1 \leq i,j \leq n \,.
\label{eq:Laplace_expansion_of_symmetric_matrix}
\end{equation}
In taking the derivatives one must take into account
that the entries are not independent. To do so, we replace
$s_{j,i} \to s_{i,j}$ with $i \leq j$ in $S$ before taking
derivatives and furthermore replace $\frac{\partial (\det S)}{\partial s_{i,k}}$
with $i>k$ by $\frac{\partial (\det S)}{\partial s_{k,i}}$.

We will now apply the identity \eqref{eq:Laplace_expansion_of_symmetric_matrix}
to the Gram matrix $S$ in eq.~\eqref{eq:extended_Gram_matrix}. However,
before doing so, we note that the first $E$ rows
only contain external invariants $\lambda_{i,j}$ and
entries which also appear in the last $L$ rows by symmetry of $S$. Derivatives
with respect to the $\lambda_{i,j}$ are not of interest in the problem at hand,
since for integration-by-parts identities \eqref{eq:IBP_schematic},
only derivatives with respect to the loop momenta play a role.
We therefore apply the identity \eqref{eq:Laplace_expansion_of_symmetric_matrix}
only to the last $L$ rows of $S$, from which we find,
\begin{equation}
\left[\sum_{k=1}^{E+L} (1{+}\delta_{i,k}) x_{j,k} \frac{\partial F}{\partial x_{i,k}}\right] - 2\delta_{i,j} F = 0 \,,
\label{eq:Laplace_expansion_of_Baikov_polynomial}
\end{equation}
where $E+1 \leq i \leq E+L$ and $1 \leq j \leq E+L$.
We can express the derivatives with respect to $x_{i,k}$
in terms of derivatives with respect to $z_\alpha$
by making use of the chain rule,
\begin{equation}
\frac{\partial F}{\partial x_{i,k}} =
\sum_{\alpha=1}^m \frac{\partial z_\alpha}{\partial x_{i,k}}
\frac{\partial F}{\partial z_\alpha}
\hspace{2mm} \mathrm{for} \hspace{1mm}
\left\{ \begin{array}{l}
1 \hspace{-0.7mm}\leq\hspace{-0.7mm} i \hspace{-0.7mm}\leq\hspace{-0.7mm} E \,,
\hspace{1.3mm} E{+}1 \hspace{-0.7mm}\leq\hspace{-0.7mm} k \hspace{-0.7mm}\leq\hspace{-0.7mm} E{+}L \,, \\[2mm]
E{+}1 \leq i \leq k \leq E{+}L \,.
\end{array} \right.
\label{eq:chain_rule}
\end{equation}
By splitting the sum in eq.~\eqref{eq:Laplace_expansion_of_Baikov_polynomial}
into sums over the first $E$, subsequent $i-1-E$ and $E+L-i+1$
terms and using that $x_{i,k} = x_{k,i}$ for the former two,
application of the chain rule \eqref{eq:chain_rule} yields,
\begin{equation}
\sum_{\alpha=1}^m (a_{i,j})_\alpha \frac{\partial F}{\partial z_\alpha} + b_{i,j} F = 0 \,,
\end{equation}
where $a_{i,j}$ and $b$ are given by the following expressions,
\begin{equation}
(a_{i,j})_\alpha =  \sum_{k=1}^{E+L} (1 + \delta_{i,k}) \frac{\partial z_\alpha}{\partial x_{i,k}} x_{j,k}
\hspace{5mm} \mathrm{and} \hspace{5mm}
b_{i,j}          = -2\delta_{i,j} \,,
\label{eq:syzygy_generators_components}
\end{equation}
for $E+1 \leq i \leq E+L$, $1 \leq j \leq E+L$ and $1 \leq \alpha \leq m$.
We conclude that
\begin{equation}
t_{i,j} = \Big((a_{i,j})_1, \ldots, (a_{i,j})_m, b_{i,j} \Big) \,,
\label{eq:syzygy_generators}
\end{equation}
with $a_{i,j}$ and $b_{i,j}$ given in eq.~\eqref{eq:syzygy_generators_components}
are solutions of eq.~\eqref{eq:syzygy_equation} in the case $c=0$.

We note that it follows from the relations
in eqs.~\eqref{eq:relation_of_z_to_x}--\eqref{eq:definition_of_z}
that the derivatives $\frac{\partial z_\alpha}{\partial x_{i,k}}$
are integers. Furthermore, we may use the relations to express
the $x$-variables as a linear combination of the $z$-variables.
This shows that the syzygies $t_{i,j}$ in eq.~\eqref{eq:syzygy_generators}
are at most linear polynomials in the Baikov variables $z_\alpha$.

We emphasize that the closed-form expressions in
eqs.~\eqref{eq:syzygy_generators_components}--\eqref{eq:syzygy_generators}
are valid for any number of loops and external legs. The only
quantities that depend on the graph in question
are the relations of the $z$-variables to the $x$-variables
in eqs.~\eqref{eq:relation_of_z_to_x}--\eqref{eq:definition_of_z}.
We note that the approach of using Laplacian expansion to
obtain syzygies works equally well in cases where the propagators
are massive, since the variables $x_{i,j}$ in eq.~(\ref{eq:x_ij_definition})
will be independent of the internal mass parameters. These mass parameters
will appear explicitly after the linear transformation from the $x_{i,j}$
variables to the Baikov variables $z_i$. For an explicit example we refer
to Sec.~\ref{sec:example_P_double_box_with_internal_mass}.

We emphasize that the closed-form expressions allow
the construction of purely $D$-dimensional integration-by-parts identities
in cases where S-polynomial based computations of syzygies
are not feasible.
Another important aspect of the syzygies
in eqs.~\eqref{eq:syzygy_generators_components}--\eqref{eq:syzygy_generators}
is that they are of degree one. This would not be guaranteed
for the output of an $S$-polynomial-based computation of
the syzygy generators which in relevant examples (see below) turn out to have higher degrees.
Low-degree syzygies are particularly advantageous if we are interested in
imposing additional constraints on the Ansatz
for the exact form in eq.~\eqref{eq:IBP_relation_ansatz}.
For example, we may demand that no integrals with squared
propagators are encountered in the IBP identities,
\begin{equation}
a_i + b_i z_i = 0 \hspace{5mm} \mathrm{where}
\hspace{5mm} i=1,\ldots,k \,.
\label{eq:no_squared_propagators}
\end{equation}
Namely, we can obtain solutions of eqs.~\eqref{eq:syzygy_equation} and \eqref{eq:no_squared_propagators}
by taking the module intersection of the module of the syzygies
in eqs.~\eqref{eq:syzygy_generators_components}--\eqref{eq:syzygy_generators},
\begin{equation}
\mathcal{T} = \big\langle t_{i,j} \hspace{1.5mm} \big| \hspace{2mm} E{+}1 \leq i \leq E{+}L
\hspace{3mm} \mathrm{and} \hspace{3mm}  1\leq j \leq E{+}L \big\rangle \,,
\label{eq:module_of_syzygies}
\end{equation}
and the module,
\begin{equation}
\mathcal{L} = \langle z_1 e_1, \ldots, z_k e_k, e_{k+1}, \ldots, e_m \rangle \,.
\end{equation}
That is, the generators of $\mathcal{T} \cap \mathcal{L}$ form a
generating set of solutions of eqs.~\eqref{eq:syzygy_equation} and \eqref{eq:no_squared_propagators} \cite{Zhang:2016kfo}.
The fact that the syzygies in eqs.~\eqref{eq:syzygy_generators_components}--\eqref{eq:syzygy_generators}
are of degree one dramatically simplifies the computation
of the module intersection $\mathcal{T} \cap \mathcal{L}$.
We remark that efficient methods for computing module intersections
are presented in ref.~\cite{Boehm:2018fpv}, and that in this reference non-trivial computations
are carried out using these methods for non-planar multi-scale diagrams.

In Sec.~\ref{sec:proof_of_completeness_of_syzygies} we give a proof
that the $L(L+E)$ syzygies in eq.~\eqref{eq:syzygy_generators}
form a generating set.

%%%%%%%%%%%%%%%%%%%%%%%%%%%%%%%%%%%%%%%%%%%%%%%%%%%%%%%%%%%%%%%%%%%%%%%%%%%%%%%%
\subsection{On-shell case}
%%%%%%%%%%%%%%%%%%%%%%%%%%%%%%%%%%%%%%%%%%%%%%%%%%%%%%%%%%%%%%%%%%%%%%%%%%%%%%%%

We now turn to obtaining a generating set of syzygies of
eq.~\eqref{eq:syzygy_equation} for a generic cut
$\mathcal{S}_\mathrm{cut} = \{ \zeta_1, \ldots, \zeta_c \}$.
We start by taking the module of the syzygies in
eq.~\eqref{eq:module_of_syzygies} and evaluating this on the cut $\mathcal{S}_\mathrm{cut}$,
\begin{equation}
\widehat{\mathcal{T}} \hspace{1mm}=\hspace{1mm} \mathcal{T} \big|_{z_q=0, \hspace{1mm} q \in \mathcal{S}_\mathrm{cut} } \,.
\end{equation}
Now, the generators $\widehat{t}_{i,j}$ of $\widehat{\mathcal{T}}$ will not in general
be solutions of eq.~\eqref{eq:syzygy_equation} because
the $\zeta_n$-entries of $\widehat{t}_{i,j}$ may be nonzero on the cut.

This leads us to consider the module,
\begin{equation}
\mathcal{Z} = \langle e_{r_1}, \ldots, e_{r_{m-c}} \rangle \,,
\end{equation}
where $e_{r_i}$ is an $(m+1)$-dimensional unit vector with $1$ in
the $r_i$ entry and $0$ elsewhere. Namely, the generators of the intersection
$\widehat{\mathcal{T}} \cap \mathcal{Z}$ are solutions of
eq.~\eqref{eq:syzygy_equation}.

The module intersection can be found with {\sc Singular}
and in practice takes less than a second to compute.

%%%%%%%%%%%%%%%%%%%%%%%%%%%%%%%%%%%%%%%%%%%%%%%%%%%%%%%%%%%%%%%%%%%%%%%%%%%%%%%%
\section{Proof of completeness of syzygies}\label{sec:proof_of_completeness_of_syzygies}
%%%%%%%%%%%%%%%%%%%%%%%%%%%%%%%%%%%%%%%%%%%%%%%%%%%%%%%%%%%%%%%%%%%%%%%%%%%%%%%%

In this section we show that the $L(L+E)$ syzygies in eqs.~\eqref{eq:syzygy_generators_components}--\eqref{eq:syzygy_generators}
form a generating set of syzygies of eq.~\eqref{eq:syzygy_equation}.
In order to give a formal proof of this fact we adopt a more general setup considering polynomial logarithmic vector fields along determinants of generic (symmetric) square matrices.
We reduce the problem to known resolutions of ideals of submaximal minors of such matrices.

Fix a field $\KK$.
For $0\ne m\in\NN$ denote by $Y=\KK^m$ affine $m$-space.
The coordinate ring of $Y$ is a polynomial ring
\begin{equation}\label{14}
\O=\O_Y=\KK[y_1,\dots,y_m] \,.
\end{equation}
Note that its group of units is $\O^*=\KK^*=\KK\setminus\set{0}$.
Since $\O$ is a Cohen--Macaulay ring, the grade or depth of any ideal of $\O$ equals its height or codimension
(cf.~Cor.~2.1.4 and Thm.~2.1.9 of ref.~\cite{BH93}).

The polynomial vector fields on $Y$ form a free $\O$-module (cf.~Prop.~16.1 of ref.~\cite{MR1322960})
\[
\Theta=\Theta_Y=\Der_\KK(\O)=\bigoplus_{i=1}^m\O\frac\p{\p y_{i}} \,.
\]
A polynomial function
\[
f\colon Y=\KK^m\to\KK
\]
is given by an element $f\in\O$.
The $\O$-submodule of $\Theta$ of logarithmic vector fields along the divisor $(f)$ is defined by (classically for squarefree $f\in\O\setminus\O^*$ and $\KK=\CC$, cf.~Sec.~1 of ref.~\cite{Sai80})
\begin{equation}\label{15}
\Der(-\log(f))=\set{\delta\in\Theta\xmid\delta(f)\in\O f}\subset\Theta_Y \,.
\end{equation}
We denote the ideal of partial derivatives of $f$ by
\[
\J_f=\ideal{\frac{\p f}{\p y_1},\dots,\frac{\p f}{\p y_m}} \,.
\]
Then $\Der(-\log(f))$ identifies with the projection to the first $m$ components of the syzygy module (cf.~eq.~\eqref{eq:syzygy_equation})
\[
\syz\left(\frac{\p f}{\p y_1},\dots,\frac{\p f}{\p y_m},f\right)\cong\syz(\J_f+\O f) \,.
\]
We call $\chi\in\Der(-\log(f))$ an \emph{Euler vector field} for $f$ if
\[
\chi(f)\in\O^*f=\KK^*f \,.
\]
If $f$ admits an Euler vector field, then
\begin{equation}\label{2}
\Der(-\log(f))=\O\chi+\Ann_{\Theta}(f) \,,
\end{equation}
where
\begin{equation}\label{6}
\Ann_{\Theta}(f)=\set{\delta\in\Theta\xmid\delta(f)=0}\cong\syz(\J_f)
\end{equation}
is the annihilator of $f$ in $\Theta$ and can be identified with the syzygy module of $\J_f$.

In the remainder of the section we specialize to the case where $f$ is the determinant of a (symmetric) $n\times n$ matrix where $0\ne n\in\NN$.
We write $\Mat_n(\O)$ for the $\O$-module of $n\times n$ matrices with entries $\O$, $\Sym_n(\O)$ (and $\Skw_n(\O)$) for its submodule of (skew-)symmetric matrices and
\[
\Sl_n(\O)=\ker(\tr\colon\Mat_n(\O)\to \O)
\]
the kernel of the trace map.
For any $A=(a_{i,j})\in\Mat_n(\O)$ denote by $A^\star=(a^\star_{i,j})\in\Mat_n(\O)$ its adjoint matrix and by
\[
I_{n-1}(A)=\ideal{a^\star_{i,j}\xmid1\le i\le j\le n}\lhd \O
\]
its ideal of submaximal minors.

Consider the determinant functions
\[
\det\colon X=\Mat_n(\KK)\to\KK \,,\quad {\det}'\colon X'=\Sym_n(\KK)\to\KK \,.
\]
The preceding discussion applies to both these cases.
The coordinate rings of $X$ and $X'$ are polynomial rings
\begin{align*}
\O_X&=\KK[x_{i,j}\mid 1\le i,j\le n] \,,\\
\O_{X'}&=\KK[x_{i,j}\mid 1\le i\le j\le n]=\O_X/\ideal{x_{i,j}-x_{j,i}} \,.
\end{align*}
The modules of polynomial vector fields on $X$ and $X'$ respectively are
\begin{align*}
\Theta_X&=\Der_\KK(\O_X)=\bigoplus_{i,j}\O_X\frac\p{\p x_{i,j}} \,,\\
\Theta_{X'}&=\Der_\KK(\O_{X'})=\bigoplus_{i\le j}\O_X\frac\p{\p x_{i,j}} \,.
\end{align*}
The following result provides generators of the modules of logarithmic vector fields along $f=\det$ and $f={\det}'$ (cf.~eq.~\eqref{15}).

Denote by $M=(x_{i,j})\in\Mat_n(\O_X)$ the \emph{generic} $n\times n$ matrix and by $S=(x_{i,j})\in\Sym_n(\O_{X'})$ its symmetric counterpart.
Note that
\[
\det=\det M\,, \quad {\det}'=\det S \,.
\]
Assume from now on that $\KK$ has characteristic different from $2$ (which will be the case in our applications).
Goryunov and Mond made the following observation (cf.~Secs.~3.1-3.2 of ref.~\cite{GM05}).

%%%%%%%%%%%%%%%%%%%%%%%%%%%%%%%%%%%%%%%%%%%%%%%%%%%%%%%%%%%%%%%%%%%%%%%%%%%%%%%%

\begin{prp}\label{5}
There are surjective maps
\[
\xymatrix@R=0em{
\Mat_n(\O_X)^{\oplus2}\ar@{->>}[r]^-\pi & \Der(-\log(\det)) \,,\\
(A,B)\ar@{|->}[r] & \sum_{i,j}(MA-BM)_{i,j}\frac{\p}{\p x_{i,j}} \,,\\
\Mat_n(\O_{X'})\ar@{->>}[r]^-{\pi'} & \Der(-\log({\det}')) \,,\\
A\ar@{|->}[r] & \sum_{i\le j}(SA+A^tS)_{i,j}\frac{\p}{\p x_{i,j}} \,.
}
\]
\end{prp}

\begin{proof}
Since
\begin{equation}\label{1}
\frac{\p\det}{\p x_{i,j}}=m^\star_{j,i} \,,\quad
\frac{\p{\det}'}{\p x_{i,j}}=(2-\delta_{i,j})s^\star_{i,j} \,,
\end{equation}
we have
\begin{equation}\label{3}
\J_{\det}=I_{n-1}(M) \,,\quad \J_{{\det}'}=I_{n-1}(S) \,,
\end{equation}
and, by Laplace expansion, $\pi$ and $\pi'$ map to the given target.
Since both $\det$ and ${\det}'$ are homogeneous, they admit standard Euler vector fields
\[
\chi=\sum_{i,j}x_{i,j}\frac\p{\p x_{i,j}} \,,\quad
\chi'=\sum_{i\le j}x_{i,j}\frac\p{\p x_{i,j}} \,.
\]
Note that
\begin{equation}\label{4}
\pi((\delta_{i,j}),0)=\chi \,,\quad\pi'((\delta_{i,j}))=2\chi \,.
\end{equation}
By Gulliksen--Neg{\aa}rd~\cite{GN72} and J{\'o}zefiak~\cite{Joz78} respectively, there are exact sequences
\[
\xymatrix@R=0em@C=2em{
\Sl_n(\O_X)^{\oplus2}\ar[r] & \Mat_n(\O_X)\ar[r] & I_{n-1}(M)\ar[r] & 0 \,,\\
(A,B)\ar@{|->}[r] & MA-BM \,,\\
& C\ar@{|->}[r] & \tr(M^\star C) \,,\\
\Sl_n(\O_{X'})\ar[r] & \Sym_n(\O_{X'})\ar[r] & I_{n-1}(S)\ar[r] & 0 \,,\\
A\ar@{|->}[r] & SA+A^tS \,,\\
& D\ar@{|->}[r] & \tr(S^\star D) \,,
}
\]
where, using eq.~\eqref{1},
\begin{align}\label{7}
\begin{split}
\tr(M^\star C)&=\sum_{i,j}m^\star_{j,i}c_{i,j}=\sum_{i,j}c_{i,j}\frac{\p\det}{\p x_{i,j}} \,,\\
\tr(S^\star D)&=\sum_{i,j}s^\star_{i,j}d_{i,j}
=\sum_{i\le j}(2-\delta_{i,j})s^\star_{i,j}d_{i,j}\\
&=\sum_{i\le j}d_{i,j}\frac{\p{\det}'}{\p x_{i,j}} \,.
\end{split}
\end{align}
Using eqs.~\eqref{6} and \eqref{3} this means that
\begin{align*}
\pi(\Sl_n(\O_X)^{\oplus2})&=\Ann_{\Theta_X}(\det) \,,\\
\pi'(\Sl_n(\O_{X'}))&=\Ann_{\Theta_{X'}}({\det}') \,.
\end{align*}
With eqs.~\eqref{2} and \eqref{4} surjectivity of $\pi$ and $\pi'$ follows.
\end{proof}

%%%%%%%%%%%%%%%%%%%%%%%%%%%%%%%%%%%%%%%%%%%%%%%%%%%%%%%%%%%%%%%%%%%%%%%%%%%%%%%%

Dropping the genericity hypothesis, we now consider polynomial matrix families
\[
\xymatrix@R=0em{
Y=\KK^m\ar[r]^-M & \Mat_n(\KK)=X \,,\\
Y=\KK^m\ar[r]^-S & \Sym_n(\KK)=X' \,.
}
\]
They are defined by matrices $M\in\Mat_n(\O)$ and $S\in\Sym_n(\O)$ with entries in $\O=\O_Y$ (cf.~eq.~\eqref{14}).
By abuse of notation, we set
\begin{align*}
\J_M&=\ideal{\frac{\p M}{\p y_k}\xmid k=1,\dots,m}\subset\Mat_n(\O) \,,\\
\J_S&=\ideal{\frac{\p S}{\p y_k}\xmid k=1,\dots,m}\subset\Sym_n(\O) \,.\\
\end{align*}
Consider the (truncated) Gulliksen--Neg{\aa}rd and J{\'o}zefiak complexes
\begin{align}\label{8}
\begin{split}
\xymatrix@R=0em@C=2em{
\Sl_n(\O)^{\oplus2}\ar[r]^-\rho & \Mat_n(\O)\ar[r] & I_{n-1}(M)\ar[r] & 0 \,,\\
\Sl_n(\O)\ar[r]^-{\rho'} & \Sym_n(\O)\ar[r] & I_{n-1}(S)\ar[r] & 0 \,.
}
\end{split}
\end{align}

%%%%%%%%%%%%%%%%%%%%%%%%%%%%%%%%%%%%%%%%%%%%%%%%%%%%%%%%%%%%%%%%%%%%%%%%%%%%%%%%

\begin{prp}\label{10}\
\begin{asparaenum}[(a)]

\item\label{10a} If $I_{n-1}(M)$ has (the maximal) codimension $4$, then there is a surjective map
\[
\xymatrix@R=0em@C=0.5em{
\rho^{-1}(\hspace{-0.2mm}\J_M \hspace{-0.5mm})\ar@{->>}[rr]^-\pi && \Ann_{\Theta}(\det M),\\
(A,B)\ar@{|->}[r] & MA-BM=\sum_{k=1}^mc_k\frac{\p M}{\p y_k}\ar@{|->}[r] & \sum_{k=1}^mc_k\frac{\p }{\p y_k} \,.
}
\]
In particular, if $\det M$ admits an Euler vector field $\chi\in\Theta$, then $\Der(-\log(\det M))$ is generated by $\chi$ and the image of $\pi$.

\item\label{10b} If $I_{n-1}(S)$ has (the maximal) codimension $3$, then there is a surjective map
\[
\xymatrix@R=0em@C=0.5em{
\rho'^{-1}(\hspace{-0.2mm}\J_S \hspace{-0.5mm})\ar@{->>}[rr]^-{\pi'} && \Ann_{\Theta}(\det S) \,,\\
A\ar@{|->}[r] & SA+A^tS=\sum_{k=1}^mc_k\frac{\p S}{\p y_k}\ar@{|->}[r] & \sum_{k=1}^mc_k\frac{\p }{\p y_k} \,.
}
\]
In particular, if $\det S$ admits an Euler vector field $\chi\in\Theta$, then $\Der(-\log(\det S))$ is generated by $\chi$ and the image of $\pi'$.

\end{asparaenum}
\end{prp}

\begin{proof}
Using eq.~\eqref{7}, the chain rule yields
\begin{align}\label{16}
\begin{split}
\tr\Big(M^\star\frac{\p M}{\p y_k}\Big)&=\sum_{i,j}\frac{\p m_{i,j}}{\p y_k}\frac{\p\det}{\p x_{i,j}}(M)=\frac{\p\det M}{\p y_k} \,,\\
\tr\Big(S^\star\frac{\p S}{\p y_k}\Big)&=\sum_{i\le j}\frac{\p s_{i,j}}{\p y_k}\frac{\p\det}{\p x_{i,j}}(S)=\frac{\p\det S}{\p y_k} \,.
\end{split}
\end{align}
By Gulliksen--Neg{\aa}rd \cite{GN72} and J{\'o}zefiak \cite{Joz78} respectively, the hypotheses imply that the complexes \eqref{8} are exact.
By eq.~\eqref{16} they induce exact sequences
\[
\xymatrix@R=0em{
\rho^{-1}(\J_M)\ar[r]^-\rho & \J_M\ar[r] & \J_{\det M}\ar[r] & 0 \,,\\
\rho'^{-1}(\J_S)\ar[r]^-{\rho'} & \J_S\ar[r] & \J_{\det S}\ar[r] & 0 \,.
}
\]
With eq.~\eqref{6} surjectivity of $\pi$ and $\pi'$ follows.
The particular claims are due to eq.~\eqref{2}.
\end{proof}

%%%%%%%%%%%%%%%%%%%%%%%%%%%%%%%%%%%%%%%%%%%%%%%%%%%%%%%%%%%%%%%%%%%%%%%%%%%%%%%%

Finally we specialize to the case of interest in our context.

\begin{cor}\label{13}
Assume that $S\in\Sym_n(\O)$ has a block form
\[
S=(s_{i,j})=
\begin{pmatrix}
S_{1,1} & S_{1,2}\\
S_{2,1} & S_{2,2}
\end{pmatrix} \,,
\]
where $S_{1,1}$ is constant invertible and $s_{i,j}=x_{i,j}=y_{\sigma_{i,j}}$ for $i\le j$ with $(i,j)$ in block column $2$.
Then $\Der(-\log(\det S))$ is generated by all
\[
\pi'(A)=\sum_{k=1}^mc_k\frac{\p }{\p y_k} \,,
\]
where (cf.~Proposition~\ref{10}.\eqref{10b})
\begin{equation}\label{12}
\sum_{k=1}^mc_k\frac{\p S}{\p y_k}=SA+A^tS \,, \quad
A=
\begin{pmatrix}
0 & A_{1,2}\\
0 & A_{2,2}
\end{pmatrix}
\in\Mat_n(\O) \,.
\end{equation}
\end{cor}

\begin{proof}
By Micali--Villamayor (see Lemma (1.1) of ref.~\cite{Joz78}), there is an invertible matrix $C$ such that
\[
C^tSC=
\begin{pmatrix}
S_0 & 0\\
0 & S'
\end{pmatrix} \,.
\]
The matrix $S_0$ is still constant invertible and $S'\equiv S_{2,2}$ modulo the variables $x_{i,j}$ with $(i,j)$ in block $(1,2)$.
The entries of $S'\in\Sym_{n'}(\O)$ are thus algebraically independent over the polynomial ring over $\KK$ in these variables.
By J{\'o}zefiak~(Thm.~(2.3) of ref.~\cite{Joz78}) it follows that $I_{n-1}(S)=I_{n'-1}(S')$ has codimension $3$.

For well-definedness of $\pi'$, it suffices to verify that $\pi'(A)\in\Der(-\log(\det S))$ if $A=(\delta_{i,k}\delta_{j,\ell})$.
In this case
\[
SA+A^tS=\sum_{i=1}^n(1+\delta_{i,\ell})y_{\sigma_{i,k}}\frac{\p S}{\p y_{\sigma_{i,\ell}}} \,,
\]
and hence, using eq.~\eqref{1} and Laplace expansion,
\begin{align*}
\pi'(A)(\det S)
&=\sum_{i=1}^n(1+\delta_{i,\ell})y_{\sigma_{i,k}}\frac{\p\det S}{\p y_{\sigma_{i,\ell}}}\\
&=\sum_{i=1}^n(1+\delta_{i,\ell})(2-\delta_{i,\ell})s_{i,k}s^\star_{i,\ell}
=2\delta_{k,\ell}\det S \,.
\end{align*}
Note that $\det S$ admits the Euler vector field
\[
\chi=\pi'((\delta_{i,n}\delta_{j,n}))=\sum_{i=1}^n(1+\delta_{i,n})y_{\sigma_{i,n}}\frac\p{\p y_{\sigma_{i,n}}}\,.
\]
So the hypotheses of Proposition~\ref{10}.\eqref{10b} are satisfied.

The module $\J_S$ consists of all symmetric matrices with $(1,1)$-block $0$.
Writing $A\in\Sl_n(\O)$ in block form
\[
A=
\begin{pmatrix}
A_{1,1} & A_{1,2}\\
A_{2,1} & A_{2,2}
\end{pmatrix} \,,
\]
eq.~\eqref{12} reduces to
\begin{equation}\label{11}
S_{1,1}A_{1,1}+S_{1,2}A_{2,1}+
A_{1,1}^tS_{1,1}+A_{2,1}^tS_{2,1}=0 \,.
\end{equation}
For any $W\in\Skw_n(\O)$, adding $WS$ to $A$ leaves $SA+A^tS$ invariant.
Using
\[
W=
\begin{pmatrix}
-A_{1,1}S_{1,1}^{-1}-S_{1,1}^{-1}A_{2,1}^tS_{2,1}S_{1,1}^{-1} & S_{1,1}^{-1}A_{2,1}^t\\[1.5mm]
-A_{2,1}S_{1,1}^{-1} & 0
\end{pmatrix}
\]
makes $A_{*,1}=0$ and turns eq.~\eqref{11} into $0=0$.
\end{proof}

Returning to the setup of Sec.~\ref{sec:Baikov_representation}, consider the matrix $S$ in eq.~\eqref{eq:extended_Gram_matrix} with the given block form.
Its submatrix $S_{1,1}$ is the Gram matrix $G$ in eq.~\eqref{eq:Gram_matrix_of_external_momenta} whose entries are the Mandelstam variables $\lambda_{i,j}$ in eq.~\eqref{eq:definition_of_lambda} which are treated as constants in the integration and IBP reduction.
As noted below eq.~\eqref{eq:definition_of_U}, $U=\det G$ is non-vanishing provided that $p_1, \ldots, p_E$ are linearly independent.
Let $\KK=\QQ(\lambda_{i,j})$ be the field of rational functions in the Mandelstam variables over the rational numbers.
Note that the characteristic of $\KK$ is 0, so that the above assumption on the characteristic is satisfied.%
\footnote{Note that while in the actual integration the variables corresponding to the external momenta may take real values and the remaining ones may take complex values, the IBP relations have a generating system defined over the rationals. Both the Laplace expansion and a syzygy module computation via Gr\"obner basis methods lead to such a generating system. This again simplifies the computation of solutions satisfying further constraints.}
Then $S_{1,1}$ is constant invertible and Corollary~\ref{13} applies.
As a result the $L(L+E)$ syzygies in eqs.~\eqref{eq:syzygy_generators_components}--\eqref{eq:syzygy_generators} generate all syzygies in eq.~\eqref{eq:syzygy_equation}.

%%%%%%%%%%%%%%%%%%%%%%%%%%%%%%%%%%%%%%%%%%%%%%%%%%%%%%%%%%%%%%%%%%%%%%%%%%%%%%%%
\section{Examples}\label{sec:examples}
%%%%%%%%%%%%%%%%%%%%%%%%%%%%%%%%%%%%%%%%%%%%%%%%%%%%%%%%%%%%%%%%%%%%%%%%%%%%%%%%

In this section we work out explicit expressions for
the syzygy generators presented in Sec.~\ref{sec:syzygies_from_Laplace_expansion}
for three diagrams.

\subsection{Fully massless planar double box}\label{sec:example_P_double_box}

As a simple example we consider the fully massless planar
double-box diagram shown in fig.~\ref{fig:massless_planar_DB_z_variables}.

\begin{figure}[!h]
\begin{center}
\includegraphics[angle=0, width=0.3\textwidth]{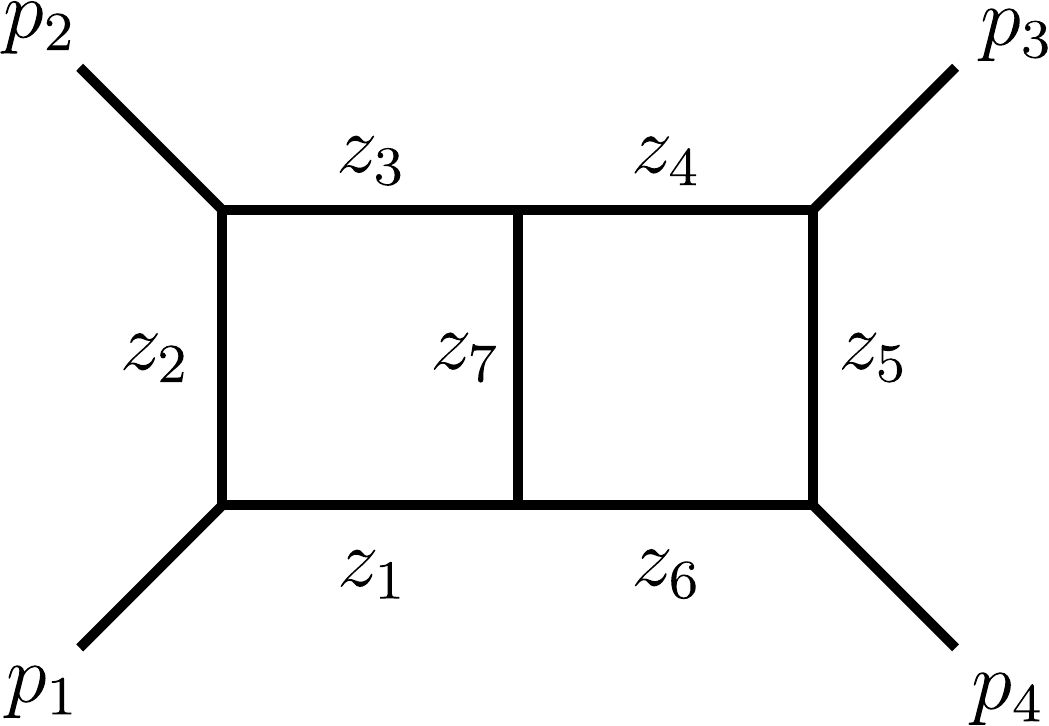}
{\vskip -0mm}
\caption{The fully massless planar double-box diagram.
All external momenta are taken to be outgoing.}
\label{fig:massless_planar_DB_z_variables}
\end{center}
\end{figure}

In this case, the combined number of propagators
and irreducible scalar products \eqref{eq:number_of_independent_loop_momentum_dot_products}
is $m = 9$. In agreement with eq.~\eqref{eq:definition_of_z},
we define the $z$-variables as follows, setting
$P_{1,2} \equiv p_1 + p_2$,
\begin{equation}
\begin{alignedat}{3}
z_1&=\ell_1^2\,,           \hspace{3mm}  && z_2 = (\ell_1 - p_1)^2\,,  \hspace{3mm} && z_3 = (\ell_1 - P_{1,2})^2\,, \\
z_4&=(\ell_2+P_{1,2})^2\,,  \hspace{3mm}  && z_5 = (\ell_2 - p_4)^2\,,  \hspace{3mm} && z_6 = \ell_2^2\,, \\
z_7&=(\ell_1+\ell_2)^2\,,  \hspace{3mm}  && z_8 = (\ell_1 + p_4)^2\,,  \hspace{3mm} && z_9 = (\ell_2 + p_1)^2\,.
\end{alignedat}
\end{equation}
We choose as the set \eqref{eq:list_of_momenta} of
all independent external and loop momenta $V = (p_1, p_2, p_3, \ell_1, \ell_2)$.
The lexicographically-ordered set of elements $(i,j)$
in eq.~\eqref{eq:independent_loop_containing_entries} then becomes,
\begin{align}
(x_1, \ldots, x_9) &\equiv
(x_{1,4}, \hspace{0.7mm} x_{1,5}, \hspace{0.7mm} x_{2,4},
\hspace{0.7mm} x_{2,5}, \hspace{0.7mm} x_{3,4}, \hspace{0.7mm} x_{3,5}, \nonumber \\
& \hspace{20mm} x_{4,4}, \hspace{0.7mm} x_{4,5}, \hspace{0.7mm} x_{5,5}) \,,
\label{eq:list_of_x_vars_massless_db}
\end{align}
and it immediately follows that the matrices in
eq.~\eqref{eq:relation_of_z_to_x} are given by,
\begin{equation}
A = \begin{pmatrix}
 0 & 0 & 0 & 0 & 0 & 0 & 1 & 0 & 0 \\
 -2 & 0 & 0 & 0 & 0 & 0 & 1 & 0 & 0 \\
 -2 & 0 & -2 & 0 & 0 & 0 & 1 & 0 & 0 \\
 0 & 2 & 0 & 2 & 0 & 0 & 0 & 0 & 1 \\
 0 & 2 & 0 & 2 & 0 & 2 & 0 & 0 & 1 \\
 0 & 0 & 0 & 0 & 0 & 0 & 0 & 0 & 1 \\
 0 & 0 & 0 & 0 & 0 & 0 & 1 & 2 & 1 \\
 -2 & 0 & -2 & 0 & -2 & 0 & 1 & 0 & 0 \\
 0 & 2 & 0 & 0 & 0 & 0 & 0 & 0 & 1 \\
\end{pmatrix} \,,
\label{eq:A-matrix_massless_db}
\end{equation}
and, for $\alpha = 1, \ldots, 9$,
\begin{equation}
B_\alpha = 0 \hspace{3mm} \mathrm{for} \hspace{2mm} \alpha \notin \{ 3,4 \}
\hspace{3mm} \mathrm{and} \hspace{3mm}
B_3 = B_4 = \begin{pmatrix} 0 & 1 & 0 \\ 1 & 0 & 0 \\ 0 & 0 & 0 \end{pmatrix} \,,
\label{eq:B-matrices_massless_db}
\end{equation}
and $m_\alpha = 0$.
We can now use eqs.~\eqref{eq:relation_of_z_to_x}--\eqref{eq:definition_of_z}
to express the syzygy generators in
eqs.~\eqref{eq:syzygy_generators_components}--\eqref{eq:syzygy_generators}
in terms of the $z_\alpha$, yielding,
\begin{align}
t_{4,1} &= (z_1{-}z_2, \hspace{0.4mm} z_1{-}z_2, \hspace{0.4mm} {-}s{+}z_1{-}z_2, \hspace{0.4mm} 0, \hspace{0.4mm} 0, \hspace{0.4mm} 0, \hspace{0.4mm} \nonumber \\
& \pushright{ z_1{-}z_2{-}z_6{+}z_9, \hspace{0.4mm} t{+}z_1{-}z_2, \hspace{0.4mm} 0, \hspace{0.4mm} 0) \,, } \nonumber \\[1.5mm]
t_{4,2} &= (s{+}z_2{-}z_3, \hspace{0.4mm} z_2{-}z_3, \hspace{0.4mm} z_2{-}z_3, \hspace{0.4mm} 0, \hspace{0.4mm} 0, \hspace{0.4mm} 0, \hspace{0.4mm} \nonumber \\
& \pushright{ z_2{-}z_3{+}z_4{-}z_9, \hspace{0.4mm} {-}t{+}z_2{-}z_3, \hspace{0.4mm} 0, \hspace{0.4mm} 0) \,, } \nonumber \\[1.5mm]
t_{4,3} &= ({-}s{+}z_3{-}z_8, \hspace{0.4mm} t{+}z_3{-}z_8, \hspace{0.4mm} z_3{-}z_8, \hspace{0.4mm} 0, \hspace{0.4mm} 0, \hspace{0.4mm} 0, \hspace{0.4mm} \nonumber \\
& \pushright{ z_3{-}z_4{+}z_5{-}z_8, \hspace{0.4mm} z_3{-}z_8, \hspace{0.4mm} 0, \hspace{0.4mm} 0) \,, } \nonumber \\[1.5mm]
t_{4,4} &= (2 z_1, \hspace{0.4mm} z_1{+}z_2, \hspace{0.4mm} {-}s{+}z_1{+}z_3, \hspace{0.4mm} 0, \hspace{0.4mm} 0, \hspace{0.4mm} 0, \hspace{0.4mm} \nonumber \\
& \pushright{ z_1{-}z_6{+}z_7, \hspace{0.4mm} z_1{+}z_8, \hspace{0.4mm} 0, \hspace{0.4mm} {-}2) \,, } \nonumber \\[1.5mm]
t_{4,5} &= ({-}z_1{-}z_6{+}z_7, \hspace{0.4mm} {-}z_1{+}z_7{-}z_9, \hspace{0.4mm} s{-}z_1{-}z_4{+}z_7, \nonumber \\
       &\hspace{10mm} 0, \hspace{0.4mm} 0, \hspace{0.4mm} 0, {-}z_1{+}z_6{+}z_7, \hspace{0.4mm} {-}z_1{-}z_5{+}z_7, \hspace{0.4mm} 0, \hspace{0.4mm} 0) \,, \label{eq:syzygies_massless_planar_db} \\[1.5mm]
t_{5,1} &= (0, \hspace{0.4mm} 0, \hspace{0.4mm} 0, \hspace{0.4mm} s{-}z_6{+}z_9, \hspace{0.4mm} {-}t{-}z_6{+}z_9, \hspace{0.4mm} z_9{-}z_6, \nonumber \\
& \pushright{ z_1{-}z_2{-}z_6{+}z_9, \hspace{0.4mm} 0, \hspace{0.4mm} z_9{-}z_6, \hspace{0.4mm} 0) \,, } \nonumber \\[1.5mm]
t_{5,2} &= (0, \hspace{0.4mm} 0, \hspace{0.4mm} 0, \hspace{0.4mm} z_4{-}z_9, \hspace{0.4mm} t{+}z_4{-}z_9, \hspace{0.4mm} {-}s{+}z_4{-}z_9, \nonumber \\
& \pushright{ z_2{-}z_3{+}z_4{-}z_9, \hspace{0.4mm} 0, \hspace{0.4mm} z_4{-}z_9, \hspace{0.4mm} 0) \,, } \nonumber \\[1.5mm]
t_{5,3} &= (0, \hspace{0.4mm} 0, \hspace{0.4mm} 0, \hspace{0.4mm} z_5{-}z_4, \hspace{0.4mm} z_5{-}z_4, \hspace{0.4mm} s{-}z_4{+}z_5, \nonumber \\
& \pushright{ z_3{-}z_4{+}z_5{-}z_8, \hspace{0.4mm} 0, \hspace{0.4mm} {-}t{-}z_4{+}z_5, \hspace{0.4mm} 0) \,, } \nonumber \\[1.5mm]
t_{5,4} &= (0, \hspace{0.4mm} 0, \hspace{0.4mm} 0, \hspace{0.4mm} s{-}z_3{-}z_6{+}z_7, \hspace{0.4mm} {-}z_6{+}z_7{-}z_8, \hspace{0.4mm} {-}z_1{-}z_6{+}z_7, \nonumber \\
& \pushright{ z_1{-}z_6{+}z_7, \hspace{0.4mm} 0, \hspace{0.4mm} {-}z_2{-}z_6{+}z_7, \hspace{0.4mm} 0) \,, } \nonumber \\[1.5mm]
t_{5,5} &= (0, \hspace{0.4mm} 0, \hspace{0.4mm} 0, \hspace{0.4mm} {-}s{+}z_4{+}z_6, \hspace{0.4mm} z_5{+}z_6, \hspace{0.4mm} 2 z_6, \nonumber \\
& \pushright{ {-}z_1{+}z_6{+}z_7, \hspace{0.4mm} 0, \hspace{0.4mm} z_6{+}z_9, \hspace{0.4mm} {-}2) \,. } \nonumber
\end{align}

\noindent Syzygies obtained from S-polynomial-based computations are not
guaranteed to be of degree one. For example, from the {\sc Singular} command {\tt syz}
one can obtain a representation with 13 generators of up to cubic degree. More specifically,
{\tt syz} produces 10 generators of degree one, two generators of degree two,
and one generator of degree three.

Expressions for on-shell syzygies are too lengthy to record here,
but we give a few examples: on the cut $\mathcal{S}_\mathrm{cut} = \{1,4,7\}$
one can find a representation of $\widehat{\mathcal{T}} \cap \mathcal{Z}$ with
18 generators of up to cubic degree, and on the cut $\mathcal{S}_\mathrm{cut} = \{2,5,7\}$
a representation of $\widehat{\mathcal{T}} \cap \mathcal{Z}$ with
20 generators of up to cubic degree.

\subsection{Planar double box with internal mass}\label{sec:example_P_double_box_with_internal_mass}

As a more non-trivial example we consider a planar
double-box diagram with propagators of equal mass
as shown in fig.~\ref{fig:massive_planar_DB_z_variables}.

\begin{figure}[!h]
\begin{center}
\includegraphics[angle=0, width=0.3\textwidth]{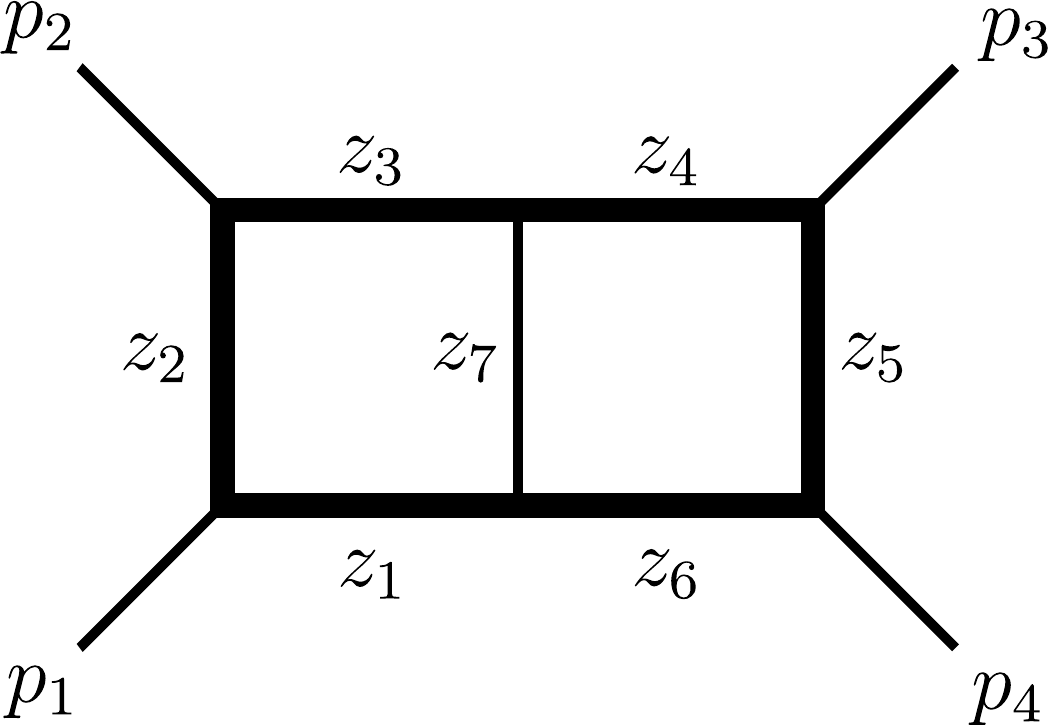}
{\vskip -0mm}
\caption{Planar double-box diagram. The bold lines represent massive
propagators with the same mass $M$. All external momenta are taken to be outgoing.}
\label{fig:massive_planar_DB_z_variables}
\end{center}
\end{figure}

As in the massless case in Sec.~\ref{sec:example_P_double_box}, the combined number of propagators
and irreducible scalar products \eqref{eq:number_of_independent_loop_momentum_dot_products}
is $m = 9$. In analogy with eq.~\eqref{eq:definition_of_z},
we define the $z$-variables as follows, setting
$P_{1,2} \equiv p_1 + p_2$,
\begin{equation}
\begin{alignedat}{3}
z_1&=\ell_1^2 - M^2\,,               \hspace{3mm}  &&  z_2 = (\ell_1 - p_1)^2 - M^2\,,    \\
z_3&= (\ell_1 - P_{1,2})^2  - M^2\,, \hspace{3mm}  &&  z_4 = (\ell_2+P_{1,2})^2  - M^2\,, \\
z_5&= (\ell_2 - p_4)^2 - M^2\,,      \hspace{3mm}  &&  z_6 = \ell_2^2  - M^2 \,,          \\
z_7&=(\ell_1+\ell_2)^2\,,            \hspace{3mm}  &&  z_8 = (\ell_1 + p_4)^2\,,          \\
z_9&= (\ell_2 + p_1)^2\,.
\end{alignedat}
\end{equation}
Again we choose as the set \eqref{eq:list_of_momenta} of
all independent external and loop momenta $V = (p_1, p_2, p_3, \ell_1, \ell_2)$.
The lexicographically-ordered set of elements $x_{i,j}$
in eq.~\eqref{eq:independent_loop_containing_entries}
is again that in eq.~\eqref{eq:list_of_x_vars_massless_db}, and
the matrices in eq.~\eqref{eq:relation_of_z_to_x} are given by
eqs.~\eqref{eq:A-matrix_massless_db} and \eqref{eq:B-matrices_massless_db},
whereas $m_\alpha$ in eq.~\eqref{eq:relation_of_z_to_x} is given by
$m_\alpha = M$ for $1 \leq \alpha \leq 6$ and $m_\beta = 0$
for $7 \leq \beta \leq 9$.

Using eqs.~\eqref{eq:relation_of_z_to_x}--\eqref{eq:definition_of_z}
to express the syzygy generators in
eqs.~\eqref{eq:syzygy_generators_components}--\eqref{eq:syzygy_generators}
in terms of the $z_\alpha$, we find in the case at hand,
\begin{align}
t_{4,1} &= (z_1 \hspace{-0.4mm} - \hspace{-0.4mm} z_2, \hspace{0.5mm} z_1 \hspace{-0.4mm} - \hspace{-0.4mm} z_2, \hspace{0.5mm} \hspace{-0.4mm} - \hspace{-0.4mm} s \hspace{-0.4mm} + \hspace{-0.4mm} z_1 \hspace{-0.4mm} - \hspace{-0.4mm} z_2, \hspace{0.5mm} 0, \hspace{0.5mm} 0, \hspace{0.5mm} 0, \nonumber \\
& \pushright{ \hspace{-0.4mm} z_1 \hspace{-0.4mm} - \hspace{-0.4mm} z_2 \hspace{-0.4mm} - \hspace{-0.4mm} z_6 \hspace{-0.4mm} + \hspace{-0.4mm} z_9 \hspace{-0.4mm} - \hspace{-0.4mm} M^2, \hspace{0.5mm} t \hspace{-0.4mm} + \hspace{-0.4mm} z_1 \hspace{-0.4mm} - \hspace{-0.4mm} z_2, \hspace{0.5mm} 0, \hspace{0.5mm} 0) \,, \nonumber } \\[1.5mm]
t_{4,2} &= (s \hspace{-0.4mm} + \hspace{-0.4mm} z_2 \hspace{-0.4mm} - \hspace{-0.4mm} z_3, \hspace{0.5mm} z_2 \hspace{-0.4mm} - \hspace{-0.4mm} z_3, \hspace{0.5mm} z_2 \hspace{-0.4mm} - \hspace{-0.4mm} z_3, \hspace{0.5mm} 0, \hspace{0.5mm} 0, \hspace{0.5mm} 0, \hspace{0.5mm} \nonumber \\
& \pushright{ z_2 \hspace{-0.4mm} - \hspace{-0.4mm} z_3 \hspace{-0.4mm} + \hspace{-0.4mm} z_4 \hspace{-0.4mm} - \hspace{-0.4mm} z_9 \hspace{-0.4mm} + \hspace{-0.4mm} M^2, \hspace{0.5mm} \hspace{-0.4mm} - \hspace{-0.4mm} t \hspace{-0.4mm} + \hspace{-0.4mm} z_2 \hspace{-0.4mm} - \hspace{-0.4mm} z_3, \hspace{0.5mm} 0, \hspace{0.5mm} 0) \,,} \nonumber \\[1.5mm]
t_{4,3} &= (- \hspace{-0.4mm} s \hspace{-0.4mm} + \hspace{-0.4mm} z_3 \hspace{-0.4mm} - \hspace{-0.4mm} z_8 \hspace{-0.4mm}+\hspace{-0.4mm} M^2, \hspace{0.5mm} t \hspace{-0.4mm} + \hspace{-0.4mm} z_3 \hspace{-0.4mm} - \hspace{-0.4mm} z_8 \hspace{-0.4mm}+\hspace{-0.4mm} M^2, \hspace{0.5mm} z_3 \hspace{-0.4mm} - \hspace{-0.4mm} z_8 \hspace{-0.4mm}+\hspace{-0.4mm} M^2, \nonumber \\
& \pushright{ \hspace{0.5mm} 0, \hspace{0.5mm} 0, \hspace{0.5mm} 0, \hspace{0.5mm} z_3 \hspace{-0.4mm} - \hspace{-0.4mm} z_4 \hspace{-0.4mm} + \hspace{-0.4mm} z_5 \hspace{-0.4mm} - \hspace{-0.4mm} z_8 \hspace{-0.4mm}+\hspace{-0.4mm} M^2, \hspace{0.5mm} z_3 \hspace{-0.4mm} - \hspace{-0.4mm} z_8 \hspace{-0.4mm} + \hspace{-0.4mm} M^2, \hspace{0.5mm} 0, \hspace{0.5mm} 0) \,,} \nonumber \\[1.5mm]
t_{4,4} &= (2z_1 \hspace{-0.4mm} + \hspace{-0.4mm} 2M^2, \hspace{0.5mm} z_1 \hspace{-0.4mm} + \hspace{-0.4mm} z_2 + \hspace{-0.4mm} 2 M^2 \hspace{-0.4mm}, -\hspace{-0.4mm} s \hspace{-0.4mm} + \hspace{-0.4mm} z_1 \hspace{-0.4mm} + \hspace{-0.4mm} z_3 \hspace{-0.4mm} + \hspace{-0.4mm} 2M^2, \nonumber \\
& \pushright{ \hspace{0.5mm} 0, \hspace{0.5mm} 0, \hspace{0.5mm} 0, \hspace{0.5mm} z_1 \hspace{-0.4mm} - \hspace{-0.4mm} z_6 \hspace{-0.4mm} + \hspace{-0.4mm} z_7, z_1 \hspace{-0.4mm} + \hspace{-0.4mm} z_8 \hspace{-0.4mm} + \hspace{-0.4mm} M^2, \hspace{0.5mm} 0, \hspace{0.5mm} \hspace{-0.4mm} - \hspace{-0.4mm} 2) \,,} \nonumber \\[1.5mm]
t_{4,5} &= (-\hspace{-0.4mm} z_1 \hspace{-0.4mm} - \hspace{-0.4mm} z_6 \hspace{-0.4mm} + \hspace{-0.4mm} z_7 \hspace{-0.4mm} - \hspace{-0.4mm} 2 M^2, \hspace{0.5mm} \hspace{-0.4mm} - \hspace{-0.4mm} z_1 \hspace{-0.4mm} + \hspace{-0.4mm} z_7 \hspace{-0.4mm} - \hspace{-0.4mm} z_9 \hspace{-0.4mm} - \hspace{-0.4mm} M^2, \nonumber \\
& \hspace{5mm} s \hspace{-0.4mm} - \hspace{-0.4mm} z_1 \hspace{-0.4mm} - \hspace{-0.4mm} z_4 \hspace{-0.4mm} + \hspace{-0.4mm} z_7 - \hspace{-0.4mm} 2 M^2, \hspace{0.5mm} 0, \hspace{0.5mm} 0, \hspace{0.5mm} 0, \hspace{0.5mm}  \hspace{-0.4mm} - \hspace{-0.4mm} z_1 \hspace{-0.4mm} + \hspace{-0.4mm} z_6 \hspace{-0.4mm} + \hspace{-0.4mm} z_7, \nonumber \\
& \hspace{20mm} - \hspace{-0.4mm} z_1 \hspace{-0.4mm} - \hspace{-0.4mm} z_5 \hspace{-0.4mm} + \hspace{-0.4mm} z_7 - \hspace{-0.4mm} 2 M^2 \hspace{-0.4mm}, \hspace{0.5mm} 0, \hspace{0.5mm} 0) \,, \\[1.5mm]
t_{5,1} &= (0, \hspace{0.5mm} 0, \hspace{0.5mm} 0, \hspace{0.5mm} \hspace{-0.4mm} s \hspace{-0.4mm} - \hspace{-0.4mm} z_6 \hspace{-0.4mm} + \hspace{-0.4mm} z_9
\hspace{-0.4mm} - \hspace{-0.4mm} M^2, \hspace{0.5mm} \hspace{-0.4mm} - \hspace{-0.4mm} t \hspace{-0.4mm} - \hspace{-0.4mm} z_6 \hspace{-0.4mm} + \hspace{-0.4mm} z_9 \hspace{-0.4mm} - \hspace{-0.4mm} M^2, \nonumber \\
& \pushright{\hspace{-4mm} z_9 \hspace{-0.4mm} - \hspace{-0.4mm} z_6 \hspace{-0.4mm} - \hspace{-0.4mm} M^2, \hspace{0.5mm} z_1 \hspace{-0.4mm} - \hspace{-0.4mm} z_2 \hspace{-0.4mm} - \hspace{-0.4mm} z_6 \hspace{-0.4mm} + \hspace{-0.4mm} z_9 \hspace{-0.4mm} - \hspace{-0.4mm} M^2, \hspace{0.5mm} 0, z_9 - \hspace{-0.4mm} z_6 \hspace{-0.4mm} - \hspace{-0.4mm} M^2, \hspace{0.5mm} 0) \,,} \nonumber \\[1.5mm]
t_{5,2} &= (0, \hspace{0.5mm} 0, \hspace{0.5mm} 0, \hspace{0.5mm} z_4 \hspace{-0.4mm} - \hspace{-0.4mm} z_9  \hspace{-0.4mm}+\hspace{-0.4mm} M^2, \hspace{0.5mm} \hspace{-0.4mm} t \hspace{-0.4mm} + \hspace{-0.4mm} z_4 \hspace{-0.4mm} - \hspace{-0.4mm} z_9 \hspace{-0.4mm}+\hspace{-0.4mm} M^2, \nonumber \\
& \pushright{\hspace{-4mm} - \hspace{-0.4mm} s \hspace{-0.4mm} + \hspace{-0.4mm} z_4 \hspace{-0.4mm} - \hspace{-0.4mm} z_9 \hspace{-0.4mm}+\hspace{-0.4mm} M^2,
\hspace{0.5mm} \hspace{-0.4mm} z_2 \hspace{-0.4mm} - \hspace{-0.4mm} z_3 \hspace{-0.4mm} + \hspace{-0.4mm} z_4 \hspace{-0.4mm} - \hspace{-0.4mm} z_9
\hspace{-0.4mm}+\hspace{-0.4mm} M^2, \hspace{0.5mm} 0,} \nonumber \\
& \pushright{ \hspace{0.5mm} z_4 \hspace{-0.4mm} - \hspace{-0.4mm} z_9 \hspace{-0.4mm}+\hspace{-0.4mm} M^2, \hspace{0.5mm} 0) \,,} \nonumber \\[1.5mm]
t_{5,3} &= (0, \hspace{0.5mm} 0, \hspace{0.5mm} 0, \hspace{0.5mm} z_5 \hspace{-0.4mm} - \hspace{-0.4mm} z_4, \hspace{0.5mm} z_5 \hspace{-0.4mm} - \hspace{-0.4mm} z_4, \hspace{0.5mm} s \hspace{-0.4mm} - \hspace{-0.4mm} z_4 \hspace{-0.4mm} + \hspace{-0.4mm} z_5, \nonumber \\
& \pushright{\hspace{0.5mm} z_3 \hspace{-0.4mm} - \hspace{-0.4mm} z_4 \hspace{-0.4mm} + \hspace{-0.4mm} z_5 \hspace{-0.4mm} - \hspace{-0.4mm} z_8 \hspace{-0.4mm} + \hspace{-0.4mm} M^2, \hspace{0.5mm} 0, \hspace{0.5mm}  \hspace{-0.4mm} - \hspace{-0.4mm} t \hspace{-0.4mm} - \hspace{-0.4mm} z_4 \hspace{-0.4mm} + \hspace{-0.4mm} z_5, \hspace{0.5mm} 0) \,,} \nonumber \\[1.5mm]
t_{5,4} &= (0, \hspace{0.5mm} 0, \hspace{0.5mm} 0, \hspace{0.5mm} s \hspace{-0.4mm} - \hspace{-0.4mm} z_3 \hspace{-0.4mm} - \hspace{-0.4mm} z_6 \hspace{-0.4mm} + \hspace{-0.4mm} z_7
\hspace{-0.4mm} - \hspace{-0.4mm} 2 M^2, \hspace{0.5mm} \hspace{-0.4mm} - \hspace{-0.4mm} z_6 \hspace{-0.4mm} + \hspace{-0.4mm} z_7 \hspace{-0.4mm} - \hspace{-0.4mm} z_8
\hspace{-0.4mm} - \hspace{-0.4mm} M^2, \nonumber \\
& \pushright{\hspace{-4mm} \hspace{-0.4mm} z_7 - \hspace{-0.4mm} z_1 \hspace{-0.4mm} - \hspace{-0.4mm} z_6 \hspace{-0.4mm} - \hspace{-0.4mm} 2 M^2, \hspace{0.5mm} \hspace{-0.4mm} z_1 \hspace{-0.4mm} - \hspace{-0.4mm} z_6 \hspace{-0.4mm} + \hspace{-0.4mm} z_7,  \hspace{0.5mm} 0,\hspace{0.5mm} \hspace{-0.4mm} z_7 \hspace{-0.4mm} - \hspace{-0.4mm} z_2 \hspace{-0.4mm} - \hspace{-0.4mm} z_6 \hspace{-0.4mm} - \hspace{-0.4mm} 2 M^2, \hspace{0.5mm} 0) \,,} \nonumber \\[1.5mm]
t_{5,5} &= (0, \hspace{0.5mm} 0, \hspace{0.5mm} 0, \hspace{0.5mm} \hspace{-0.4mm} - \hspace{-0.4mm} s \hspace{-0.4mm} + \hspace{-0.4mm} z_4 \hspace{-0.4mm} + \hspace{-0.4mm} z_6 \hspace{-0.4mm} + \hspace{-0.4mm} 2M^2, \hspace{0.5mm} z_5 \hspace{-0.4mm} + \hspace{-0.4mm} z_6 \hspace{-0.4mm} + \hspace{-0.4mm} 2M^2, \nonumber \\
& \pushright{\hspace{0.5mm} 2z_6 \hspace{-0.4mm}+\hspace{-0.4mm} 2M^2, \hspace{0.5mm} \hspace{-0.4mm} - \hspace{-0.4mm} z_1 + \hspace{-0.4mm} z_6 \hspace{-0.4mm} + \hspace{-0.4mm} z_7, \hspace{0.5mm} 0, \hspace{0.5mm} z_6 \hspace{-0.4mm} + \hspace{-0.4mm} z_9 \hspace{-0.4mm} + \hspace{-0.4mm} M^2, \hspace{0.5mm}  \hspace{-0.4mm} - \hspace{-0.4mm} 2) \,,} \nonumber
\end{align}
which agrees with eq.~(\ref{eq:syzygies_massless_planar_db}) in the case $M=0$.

\subsection{Fully massless non-planar double pentagon}\label{sec:example_NP_double_pentagon}

As a yet more non-trivial example we consider the fully massless non-planar
double-pentagon diagram shown in fig.~\ref{fig:massless_non-planar_DP_z_variables}.

\begin{figure}[!h]
\begin{center}
\includegraphics[angle=0, width=0.3\textwidth]{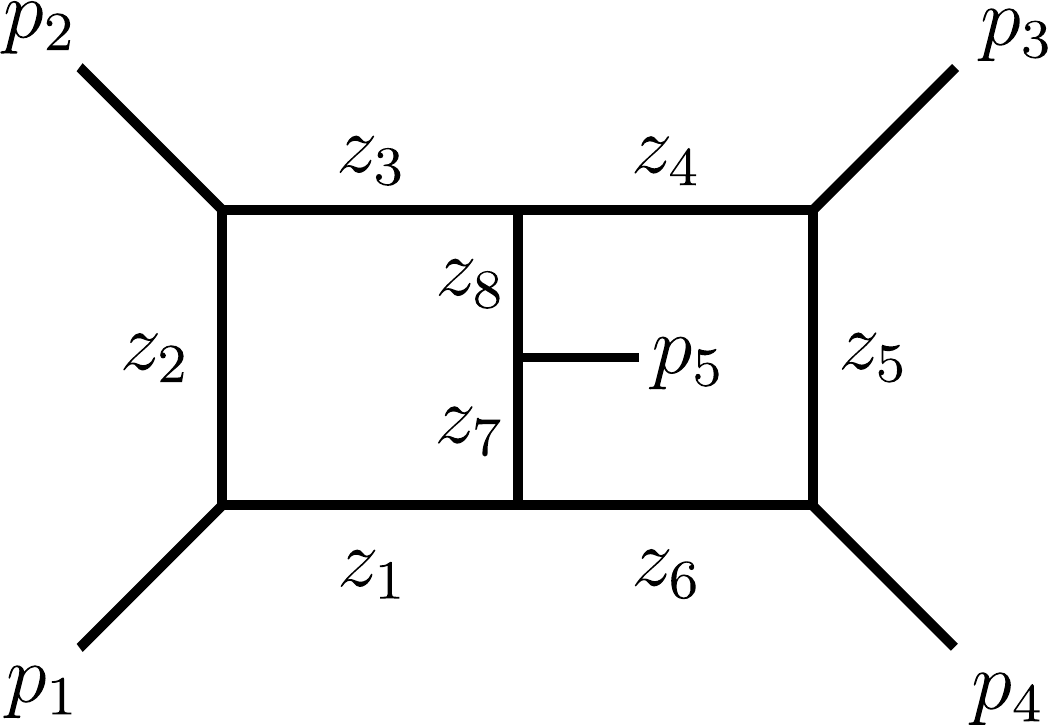}
{\vskip -0mm}
\caption{The fully massless non-planar double-pentagon diagram.
All external momenta are taken to be outgoing.}
\label{fig:massless_non-planar_DP_z_variables}
\end{center}
\end{figure}

In this case, the combined number of propagators
and irreducible scalar products \eqref{eq:number_of_independent_loop_momentum_dot_products}
is $m = 11$. In agreement with eq.~\eqref{eq:definition_of_z},
we define the $z$-variables as follows, setting
$P_{i,j} \equiv p_i + p_j$,
\begin{equation}
\begin{alignedat}{3}
   z_1&=\ell_1^2\,,           \hspace{2mm}  &&    z_2 = (\ell_1{-}p_1)^2\,,           \hspace{2mm} && z_3 = (\ell_1{-}P_{1,2})^2\,, \\
   z_4&=(\ell_2{-}P_{3,4})^2\,, \hspace{2mm}  &&    z_5 = (\ell_2{-}p_4)^2\,,           \hspace{2mm} && z_6 = \ell_2^2\,, \\
   z_7&=(\ell_1{+}\ell_2)^2\,,  \hspace{2mm}  &&    z_8 = (\ell_1{+}\ell_2{+}p_5)^2\,,  \hspace{2mm} && z_9 = (\ell_1{+}p_3)^2\,, \\
z_{10}&=(\ell_1{+}p_4)^2\,,     \hspace{2mm}  && z_{11} = (\ell_2{+}p_1)^2\,.           \hspace{2mm} &&
\end{alignedat}
\end{equation}
We choose as the set \eqref{eq:list_of_momenta} of
all independent external and loop momenta $V = (p_1, p_2, p_3, p_4, \ell_1, \ell_2)$.
The lexicographically-ordered set of elements $(i,j)$
in eq.~\eqref{eq:independent_loop_containing_entries} then becomes,
\begin{align}
(x_1, \ldots, x_{11}) &\equiv
(x_{1,5}, \hspace{0.7mm} x_{1,6}, \hspace{0.7mm} x_{2,5},
\hspace{0.7mm} x_{2,6}, \hspace{0.7mm} x_{3,5}, \hspace{0.7mm} x_{3,6}, \nonumber \\
& \hspace{14.5mm} x_{4,5}, \hspace{0.7mm} x_{4,6}, \hspace{0.7mm} x_{5,5}, \hspace{0.7mm} x_{5,6}, \hspace{0.7mm} x_{6,6}) \,,
\end{align}
and it immediately follows that the matrices in
eq.~\eqref{eq:relation_of_z_to_x} are given by,
\begin{equation}
A=\left(
\begin{array}{ccccccccccc}
 0 & 0 & 0 & 0 & 0 & 0 & 0 & 0 & 1 & 0 & 0 \\
 -2 & 0 & 0 & 0 & 0 & 0 & 0 & 0 & 1 & 0 & 0 \\
 -2 & 0 & -2 & 0 & 0 & 0 & 0 & 0 & 1 & 0 & 0 \\
 0 & 0 & 0 & 0 & 0 & -2 & 0 & -2 & 0 & 0 & 1 \\
 0 & 0 & 0 & 0 & 0 & 0 & 0 & -2 & 0 & 0 & 1 \\
 0 & 0 & 0 & 0 & 0 & 0 & 0 & 0 & 0 & 0 & 1 \\
 0 & 0 & 0 & 0 & 0 & 0 & 0 & 0 & 1 & 2 & 1 \\
 -2 & -2 & -2 & -2 & -2 & -2 & -2 & -2 & 1 & 2 & 1 \\
 0 & 0 & 0 & 0 & 2 & 0 & 0 & 0 & 1 & 0 & 0 \\
 0 & 0 & 0 & 0 & 0 & 0 & 2 & 0 & 1 & 0 & 0 \\
 0 & 2 & 0 & 0 & 0 & 0 & 0 & 0 & 0 & 0 & 1 \\
\end{array}
\right) \,,
\end{equation}
and, for $\alpha = 1, \ldots, 11$,
\begin{align}
&B_\alpha = 0 \hspace{3mm} \mathrm{for} \hspace{2mm} \alpha \notin \{ 3,4 \}
\hspace{3mm} \mathrm{and} \hspace{3mm} \\[1mm]
&
B_3 = \begin{pmatrix} 0 & 1 & 0 & 0 \\ 1 & 0 & 0 & 0 \\ 0 & 0 & 0 & 0 \\ 0 & 0 & 0 & 0 \end{pmatrix} \,, \hspace{4mm}
B_4 = \begin{pmatrix} 0 & 0 & 0 & 0 \\ 0 & 0 & 0 & 0 \\ 0 & 0 & 0 & 1 \\ 0 & 0 & 1 & 0 \end{pmatrix} \,,
\end{align}
and $m_\alpha = 0$.
We can now use eqs.~\eqref{eq:relation_of_z_to_x}--\eqref{eq:definition_of_z}
to express the syzygy generators in
eqs.~\eqref{eq:syzygy_generators_components}--\eqref{eq:syzygy_generators}
in terms of the $z_\alpha$, yielding,
\begin{align}
t_{5, 1} &= (z_{1}\hspace{-0.6mm}-\hspace{-0.6mm}z_{2}, \hspace{0.5mm} z_{1}\hspace{-0.6mm}-\hspace{-0.6mm}z_{2}, \hspace{0.5mm} \hspace{-0.6mm}-\hspace{-0.6mm}s_{1,2}\hspace{-0.6mm}+\hspace{-0.6mm}z_{1}\hspace{-0.6mm}-\hspace{-0.6mm}z_{2}, \hspace{0.5mm} 0, \hspace{0.5mm} 0, \hspace{0.5mm} 0, \hspace{0.5mm} \nonumber \\
& \hspace{-3mm} z_{1}\hspace{-0.6mm}-\hspace{-0.6mm}z_{2}\hspace{-0.6mm}-\hspace{-0.6mm}z_{6}\hspace{-0.6mm}+\hspace{-0.6mm}z_{11}, \hspace{0.5mm} \hspace{-0.6mm}-\hspace{-0.6mm}s_{1,2}\hspace{-0.6mm}-\hspace{-0.6mm}s_{1,3}\hspace{-0.6mm}-\hspace{-0.6mm}s_{1,4}\hspace{-0.6mm}+\hspace{-0.6mm}z_{1}\hspace{-0.6mm}-\hspace{-0.6mm}z_{2}\hspace{-0.6mm}-\hspace{-0.6mm}z_{6}\hspace{-0.6mm}+\hspace{-0.6mm}z_{11}, \nonumber \\
& \hspace{-3mm} s_{1,3}\hspace{-0.6mm}+\hspace{-0.6mm}z_{1}\hspace{-0.6mm}-\hspace{-0.6mm}z_{2}, \hspace{0.5mm} s_{1,4}\hspace{-0.6mm}+\hspace{-0.6mm}z_{1}\hspace{-0.6mm}-\hspace{-0.6mm}z_{2}, \hspace{0.5mm} 0, \hspace{0.5mm} 0) \,, \nonumber \\[2.5mm]
t_{5, 2} &= (s_{1,2}\hspace{-0.6mm}+\hspace{-0.6mm}z_{2}\hspace{-0.6mm}-\hspace{-0.6mm}z_{3}, \hspace{0.5mm} z_{2}\hspace{-0.6mm}-\hspace{-0.6mm}z_{3}, \hspace{0.5mm} z_{2}\hspace{-0.6mm}-\hspace{-0.6mm}z_{3}, \hspace{0.5mm} 0, \hspace{0.5mm} 0, \hspace{0.5mm} 0, \hspace{0.5mm} \nonumber \\
&\hspace{-3mm}-\hspace{-0.6mm}s_{3,4}\hspace{-0.6mm}+\hspace{-0.6mm}z_{1}\hspace{-0.6mm}+\hspace{-0.6mm}z_{2}\hspace{-0.6mm}+\hspace{-0.6mm}z_{4}\hspace{-0.6mm}+\hspace{-0.6mm}z_{7}\hspace{-0.6mm}-\hspace{-0.6mm}z_{8}\hspace{-0.6mm}-\hspace{-0.6mm}z_{9}\hspace{-0.6mm}-\hspace{-0.6mm}z_{10}\hspace{-0.6mm}-\hspace{-0.6mm}z_{11}, \nonumber \\
& \hspace{-3mm} s_{1,3}\hspace{-0.6mm}+\hspace{-0.6mm}s_{1,4}\hspace{-0.6mm}+\hspace{-0.6mm}z_{1}\hspace{-0.6mm}+\hspace{-0.6mm}z_{2}\hspace{-0.6mm}+\hspace{-0.6mm}z_{4}\hspace{-0.6mm}+\hspace{-0.6mm}z_{7}\hspace{-0.6mm}-\hspace{-0.6mm}z_{8}\hspace{-0.6mm}-\hspace{-0.6mm}z_{9}\hspace{-0.6mm}-\hspace{-0.6mm}z_{10}\hspace{-0.6mm}-\hspace{-0.6mm}z_{11}, \nonumber \\
&\hspace{-3mm} s_{1,2}\hspace{-0.6mm}+\hspace{-0.6mm}s_{2,3}\hspace{-0.6mm}+\hspace{-0.6mm}z_{2}\hspace{-0.6mm}-\hspace{-0.6mm}z_{3}, \hspace{0.5mm}\hspace{-0.6mm}-\hspace{-0.6mm}s_{1,3}\hspace{-0.6mm}-\hspace{-0.6mm}s_{1,4}\hspace{-0.6mm}-\hspace{-0.6mm}s_{2,3}\hspace{-0.6mm}-\hspace{-0.6mm}s_{3,4}\hspace{-0.6mm}+\hspace{-0.6mm}z_{2}\hspace{-0.6mm}-\hspace{-0.6mm}z_{3}, \hspace{0.5mm} 0, \hspace{0.5mm} 0) \,, \nonumber \\[2.5mm]
t_{5, 3} &= (z_{9}\hspace{-0.6mm}-\hspace{-0.6mm}z_{1}, \hspace{0.5mm} \hspace{-0.6mm}-\hspace{-0.6mm}s_{1,3}\hspace{-0.6mm}-\hspace{-0.6mm}z_{1}\hspace{-0.6mm}+\hspace{-0.6mm}z_{9}, \hspace{0.5mm} \hspace{-0.6mm}-\hspace{-0.6mm}s_{1,3}\hspace{-0.6mm}-\hspace{-0.6mm}s_{2,3}\hspace{-0.6mm}-\hspace{-0.6mm}z_{1}\hspace{-0.6mm}+\hspace{-0.6mm}z_{9}, \hspace{0.5mm} 0, \hspace{0.5mm} 0, \hspace{0.5mm} 0, \hspace{0.5mm} \nonumber \\
& \hspace{-3mm} s_{3,4}\hspace{-0.6mm}-\hspace{-0.6mm}z_{1}\hspace{-0.6mm}-\hspace{-0.6mm}z_{4}\hspace{-0.6mm}+\hspace{-0.6mm}z_{5}\hspace{-0.6mm}+\hspace{-0.6mm}z_{9}, \hspace{0.5mm} \hspace{-0.6mm}-\hspace{-0.6mm}s_{1,3}\hspace{-0.6mm}-\hspace{-0.6mm}s_{2,3}\hspace{-0.6mm}-\hspace{-0.6mm}z_{1}\hspace{-0.6mm}-\hspace{-0.6mm}z_{4}\hspace{-0.6mm}+\hspace{-0.6mm}z_{5}\hspace{-0.6mm}+\hspace{-0.6mm}z_{9}, \nonumber \\
& \hspace{-3mm} z_{9}\hspace{-0.6mm}-\hspace{-0.6mm}z_{1}, \hspace{0.5mm} s_{3,4}\hspace{-0.6mm}-\hspace{-0.6mm}z_{1}\hspace{-0.6mm}+\hspace{-0.6mm}z_{9}, \hspace{0.5mm} 0, \hspace{0.5mm} 0) \,, \nonumber \\[2.5mm]
t_{5, 4} &= (z_{10}\hspace{-0.6mm}-\hspace{-0.6mm}z_{1}, \hspace{0.5mm} \hspace{-0.6mm}-\hspace{-0.6mm}s_{1,4}\hspace{-0.6mm}-\hspace{-0.6mm}z_{1}\hspace{-0.6mm}+\hspace{-0.6mm}z_{10}, \hspace{0.5mm} -s_{1,4}\hspace{-0.6mm}-\hspace{-0.6mm}s_{2,4}\hspace{-0.6mm}-\hspace{-0.6mm}z_{1}\hspace{-0.6mm}+\hspace{-0.6mm}z_{10}, \nonumber \\
&\hspace{-3mm} 0, \hspace{0.5mm} 0, \hspace{0.5mm} 0, \hspace{0.5mm} \hspace{-0.6mm}-\hspace{-0.6mm}z_{1}\hspace{-0.6mm}-\hspace{-0.6mm}z_{5}\hspace{-0.6mm}+\hspace{-0.6mm}z_{6}\hspace{-0.6mm}+\hspace{-0.6mm}z_{10}, \nonumber \\
&\hspace{-3mm} s_{1,2}\hspace{-0.6mm}+\hspace{-0.6mm}s_{1,3}\hspace{-0.6mm}+\hspace{-0.6mm}s_{2,3}\hspace{-0.6mm}-\hspace{-0.6mm}z_{1}\hspace{-0.6mm}-\hspace{-0.6mm}z_{5}\hspace{-0.6mm}+\hspace{-0.6mm}z_{6}\hspace{-0.6mm}+\hspace{-0.6mm}z_{10}, \nonumber \\
& \hspace{-3mm} s_{3,4}\hspace{-0.6mm}-\hspace{-0.6mm}z_{1}\hspace{-0.6mm}+\hspace{-0.6mm}z_{10}, \hspace{0.5mm} z_{10}\hspace{-0.6mm}-\hspace{-0.6mm}z_{1}, \hspace{0.5mm} 0, \hspace{0.5mm} 0) \,, \nonumber \\[2.5mm]
t_{5, 5} &= (2 z_{1}, \hspace{0.5mm} z_{1}\hspace{-0.6mm}+\hspace{-0.6mm}z_{2}, \hspace{0.5mm} \hspace{-0.6mm}-\hspace{-0.6mm}s_{1,2}\hspace{-0.6mm}+\hspace{-0.6mm}z_{1}\hspace{-0.6mm}+\hspace{-0.6mm}z_{3}, \hspace{0.5mm} 0, \hspace{0.5mm} 0, \hspace{0.5mm} 0, \nonumber \\
& \hspace{-3mm} z_{1}\hspace{-0.6mm}-\hspace{-0.6mm}z_{6}\hspace{-0.6mm}+\hspace{-0.6mm}z_{7}, \hspace{0.5mm} \hspace{-0.6mm}-\hspace{-0.6mm}s_{1,2}\hspace{-0.6mm}+\hspace{-0.6mm}2 z_{1}\hspace{-0.6mm}+\hspace{-0.6mm}z_{3}\hspace{-0.6mm}-\hspace{-0.6mm}z_{6}\hspace{-0.6mm}+\hspace{-0.6mm}z_{7}\hspace{-0.6mm}-\hspace{-0.6mm}z_{9}\hspace{-0.6mm}-\hspace{-0.6mm}z_{10}, \nonumber \\
& \hspace{-3mm} z_{1}\hspace{-0.6mm}+\hspace{-0.6mm}z_{9}, \hspace{0.5mm} z_{1}\hspace{-0.6mm}+\hspace{-0.6mm}z_{10}, \hspace{0.5mm} 0, \hspace{0.5mm} \hspace{-0.6mm}-\hspace{-0.6mm}2) \,, \nonumber \\[2.5mm]
t_{5, 6} &= (\hspace{-0.6mm}-\hspace{-0.6mm}z_{1}\hspace{-0.6mm}-\hspace{-0.6mm}z_{6}\hspace{-0.6mm}+\hspace{-0.6mm}z_{7}, \hspace{0.5mm} \hspace{-0.6mm}-\hspace{-0.6mm}z_{1}\hspace{-0.6mm}+\hspace{-0.6mm}z_{7}\hspace{-0.6mm}-\hspace{-0.6mm}z_{11}, \nonumber \\
&\hspace{-3mm} s_{1,2}\hspace{-0.6mm}+\hspace{-0.6mm}s_{3,4}\hspace{-0.6mm}-\hspace{-0.6mm}2 z_{1}\hspace{-0.6mm}-\hspace{-0.6mm}z_{3}\hspace{-0.6mm}-\hspace{-0.6mm}z_{4}\hspace{-0.6mm}+\hspace{-0.6mm}z_{8}\hspace{-0.6mm}+\hspace{-0.6mm}z_{9}\hspace{-0.6mm}+\hspace{-0.6mm}z_{10}, \hspace{0.5mm} 0, \hspace{0.5mm} 0, \hspace{0.5mm} 0, \nonumber \\
& \hspace{-3mm} \hspace{-0.6mm}-\hspace{-0.6mm}z_{1}\hspace{-0.6mm}+\hspace{-0.6mm}z_{6}\hspace{-0.6mm}+\hspace{-0.6mm}z_{7}, \hspace{0.5mm} s_{1,2}\hspace{-0.6mm}-\hspace{-0.6mm}2 z_{1}\hspace{-0.6mm}-\hspace{-0.6mm}z_{3}\hspace{-0.6mm}+\hspace{-0.6mm}z_{6}\hspace{-0.6mm}+\hspace{-0.6mm}z_{8}\hspace{-0.6mm}+\hspace{-0.6mm}z_{9}\hspace{-0.6mm}+\hspace{-0.6mm}z_{10}, \nonumber \\
& \hspace{-3mm} s_{3,4}\hspace{-0.6mm}-\hspace{-0.6mm}z_{1}\hspace{-0.6mm}-\hspace{-0.6mm}z_{4}\hspace{-0.6mm}+\hspace{-0.6mm}z_{5}\hspace{-0.6mm}-\hspace{-0.6mm}z_{6}\hspace{-0.6mm}+\hspace{-0.6mm}z_{7}, \hspace{0.5mm} \hspace{-0.6mm}-\hspace{-0.6mm}z_{1}\hspace{-0.6mm}-\hspace{-0.6mm}z_{5}\hspace{-0.6mm}+\hspace{-0.6mm}z_{7}, \hspace{0.5mm} 0, \hspace{0.5mm} 0) \,, \\[2.5mm]
t_{6, 1} &= (0, \hspace{0.5mm} 0, \hspace{0.5mm} 0, \hspace{0.5mm} \hspace{-0.6mm}-\hspace{-0.6mm}s_{1,3}\hspace{-0.6mm}-\hspace{-0.6mm}s_{1,4}\hspace{-0.6mm}-\hspace{-0.6mm}z_{6}\hspace{-0.6mm}+\hspace{-0.6mm}z_{11}, \hspace{0.5mm} \hspace{-0.6mm}-\hspace{-0.6mm}s_{1,4}\hspace{-0.6mm}-\hspace{-0.6mm}z_{6}\hspace{-0.6mm}+\hspace{-0.6mm}z_{11}, \nonumber \\
& \hspace{-3mm} z_{11}\hspace{-0.6mm}-\hspace{-0.6mm}z_{6}, \hspace{0.5mm} z_{1}\hspace{-0.6mm}-\hspace{-0.6mm}z_{2}\hspace{-0.6mm}-\hspace{-0.6mm}z_{6}\hspace{-0.6mm}+\hspace{-0.6mm}z_{11}, \nonumber \\
& \hspace{-3mm} \hspace{-0.6mm}-\hspace{-0.6mm}s_{1,2}\hspace{-0.6mm}-\hspace{-0.6mm}s_{1,3}\hspace{-0.6mm}-\hspace{-0.6mm}s_{1,4}\hspace{-0.6mm}+\hspace{-0.6mm}z_{1}\hspace{-0.6mm}-\hspace{-0.6mm}z_{2}\hspace{-0.6mm}-\hspace{-0.6mm}z_{6}\hspace{-0.6mm}+\hspace{-0.6mm}z_{11}, \hspace{0.5mm} 0, \hspace{0.5mm} 0, \hspace{0.5mm} z_{11}\hspace{-0.6mm}-\hspace{-0.6mm}z_{6}, \hspace{0.5mm} 0) \,, \nonumber \\[2.5mm]
t_{6, 2} &= (0, \hspace{0.0mm} 0, \hspace{0.0mm} 0, \hspace{0.0mm} s_{1,3}\hspace{-0.8mm}+\hspace{-0.8mm}s_{1,4}\hspace{-0.8mm}+\hspace{-0.8mm}z_{1}\hspace{-0.8mm}+\hspace{-0.8mm}z_{3}\hspace{-0.8mm}+\hspace{-0.8mm}z_{4}\hspace{-0.8mm}+\hspace{-0.8mm}z_{7}\hspace{-0.8mm}-\hspace{-0.8mm}z_{8}\hspace{-0.8mm}-\hspace{-0.8mm}z_{9}\hspace{-0.8mm}-\hspace{-0.8mm}z_{10}\hspace{-0.8mm}-\hspace{-0.8mm}z_{11}, \nonumber \\
& \hspace{-3mm} s_{1,3}\hspace{-0.6mm}+\hspace{-0.6mm}s_{1,4}\hspace{-0.6mm}+\hspace{-0.6mm}s_{2,3}\hspace{-0.6mm}+\hspace{-0.6mm}z_{1}\hspace{-0.6mm}+\hspace{-0.6mm}z_{3}\hspace{-0.6mm}+\hspace{-0.6mm}z_{4}\hspace{-0.6mm}+\hspace{-0.6mm}z_{7}\hspace{-0.6mm}-\hspace{-0.6mm}z_{8}\hspace{-0.6mm}-\hspace{-0.6mm}z_{9}\hspace{-0.6mm}-\hspace{-0.6mm}z_{10}\hspace{-0.6mm}-\hspace{-0.6mm}z_{11}, \nonumber \\
& \hspace{-3mm} \hspace{-0.6mm}-\hspace{-0.6mm}s_{1,2}\hspace{-0.6mm}-\hspace{-0.6mm}s_{3,4}\hspace{-0.6mm}+\hspace{-0.6mm}z_{1}\hspace{-0.6mm}+\hspace{-0.6mm}z_{3}\hspace{-0.6mm}+\hspace{-0.6mm}z_{4}\hspace{-0.6mm}+\hspace{-0.6mm}z_{7}\hspace{-0.6mm}-\hspace{-0.6mm}z_{8}\hspace{-0.6mm}-\hspace{-0.6mm}z_{9}\hspace{-0.6mm}-\hspace{-0.6mm}z_{10}\hspace{-0.6mm}-\hspace{-0.6mm}z_{11}, \nonumber \\
&\hspace{-3mm}-\hspace{-0.6mm}s_{3,4}\hspace{-0.6mm}+\hspace{-0.6mm}z_{1}\hspace{-0.6mm}+\hspace{-0.6mm}z_{2}\hspace{-0.6mm}+\hspace{-0.6mm}z_{4}\hspace{-0.6mm}+\hspace{-0.6mm}z_{7}\hspace{-0.6mm}-\hspace{-0.6mm}z_{8}\hspace{-0.6mm}-\hspace{-0.6mm}z_{9}\hspace{-0.6mm}-\hspace{-0.6mm}z_{10}\hspace{-0.6mm}-\hspace{-0.6mm}z_{11}, \nonumber \\
&\hspace{-3mm} s_{1,3}\hspace{-0.6mm}+\hspace{-0.6mm}s_{1,4}\hspace{-0.6mm}+\hspace{-0.6mm}z_{1}\hspace{-0.6mm}+\hspace{-0.6mm}z_{2}\hspace{-0.6mm}+\hspace{-0.6mm}z_{4}\hspace{-0.6mm}+\hspace{-0.6mm}z_{7}\hspace{-0.6mm}-\hspace{-0.6mm}z_{8}\hspace{-0.6mm}-\hspace{-0.6mm}z_{9}\hspace{-0.6mm}-\hspace{-0.6mm}z_{10}\hspace{-0.6mm}-\hspace{-0.6mm}z_{11}, \hspace{0.5mm} 0, \hspace{0.5mm} 0, \nonumber \\
& \hspace{-3mm}-\hspace{-0.6mm}s_{3,4}\hspace{-0.6mm}+\hspace{-0.6mm}z_{1}\hspace{-0.6mm}+\hspace{-0.6mm}z_{3}\hspace{-0.6mm}+\hspace{-0.6mm}z_{4}\hspace{-0.6mm}+\hspace{-0.6mm}z_{7}\hspace{-0.6mm}-\hspace{-0.6mm}z_{8}\hspace{-0.6mm}-\hspace{-0.6mm}z_{9}\hspace{-0.6mm}-\hspace{-0.6mm}z_{10}\hspace{-0.6mm}-\hspace{-0.6mm}z_{11}, \hspace{0.5mm} 0) \,, \nonumber \\[2.5mm]
t_{6, 3} &= (0, \hspace{0.5mm} 0, \hspace{0.5mm} 0, \hspace{0.5mm} z_{5}\hspace{-0.6mm}-\hspace{-0.6mm}z_{4}, \hspace{0.5mm} z_{5}\hspace{-0.6mm}-\hspace{-0.6mm}z_{4}, \hspace{0.5mm} s_{3,4}\hspace{-0.6mm}-\hspace{-0.6mm}z_{4}\hspace{-0.6mm}+\hspace{-0.6mm}z_{5}, \nonumber \\
& \hspace{-3mm} s_{3,4}\hspace{-0.6mm}-\hspace{-0.6mm}z_{1}\hspace{-0.6mm}-\hspace{-0.6mm}z_{4}\hspace{-0.6mm}+\hspace{-0.6mm}z_{5}\hspace{-0.6mm}+\hspace{-0.6mm}z_{9}, \hspace{0.5mm} \hspace{-0.6mm}-\hspace{-0.6mm}s_{1,3}\hspace{-0.6mm}-\hspace{-0.6mm}s_{2,3}\hspace{-0.6mm}-\hspace{-0.6mm}z_{1}\hspace{-0.6mm}-\hspace{-0.6mm}z_{4}\hspace{-0.6mm}+\hspace{-0.6mm}z_{5}\hspace{-0.6mm}+\hspace{-0.6mm}z_{9}, \hspace{0.5mm} 0, \hspace{0.5mm} 0, \nonumber \\
& \hspace{-3mm} s_{1,3}\hspace{-0.6mm}+\hspace{-0.6mm}s_{3,4}\hspace{-0.6mm}-\hspace{-0.6mm}z_{4}\hspace{-0.6mm}+\hspace{-0.6mm}z_{5}, \hspace{0.5mm} 0) \,, \nonumber \\[2.5mm]
t_{6, 4} &= (0, \hspace{0.5mm} 0, \hspace{0.5mm} 0, \hspace{0.5mm} \hspace{-0.6mm}-\hspace{-0.6mm}s_{3,4}\hspace{-0.6mm}-\hspace{-0.6mm}z_{5}\hspace{-0.6mm}+\hspace{-0.6mm}z_{6}, \hspace{0.5mm} z_{6}\hspace{-0.6mm}-\hspace{-0.6mm}z_{5}, \hspace{0.5mm} z_{6}\hspace{-0.6mm}-\hspace{-0.6mm}z_{5}, \nonumber \\
& \hspace{-3mm} -\hspace{-0.8mm}z_{1}\hspace{-0.8mm}-\hspace{-0.8mm}z_{5}\hspace{-0.8mm}+\hspace{-0.8mm}z_{6}\hspace{-0.8mm}+\hspace{-0.8mm}z_{10}, \hspace{0.5mm} s_{1,2}\hspace{-0.8mm}+\hspace{-0.8mm}s_{1,3}\hspace{-0.8mm}+\hspace{-0.8mm}s_{2,3}\hspace{-0.8mm}-\hspace{-0.8mm}z_{1}\hspace{-0.8mm}-\hspace{-0.8mm}z_{5}\hspace{-0.8mm}+\hspace{-0.8mm}z_{6}\hspace{-0.8mm}+\hspace{-0.8mm}z_{10}, \hspace{0.5mm} 0, \hspace{0.5mm} 0, \nonumber \\
&\hspace{-3mm} s_{1,4}\hspace{-0.6mm}-\hspace{-0.6mm}z_{5}\hspace{-0.6mm}+\hspace{-0.6mm}z_{6}, \hspace{0.5mm} 0) \,, \nonumber \\[2.5mm]
t_{6, 5} &= (0, \hspace{0.5mm} 0, \hspace{0.5mm} 0, \hspace{0.5mm} z_{1}\hspace{-0.6mm}-\hspace{-0.6mm}z_{6}\hspace{-0.6mm}+\hspace{-0.6mm}z_{7}\hspace{-0.6mm}-\hspace{-0.6mm}z_{9}\hspace{-0.6mm}-\hspace{-0.6mm}z_{10}, \hspace{0.5mm} \hspace{-0.6mm}-\hspace{-0.6mm}z_{6}\hspace{-0.6mm}+\hspace{-0.6mm}z_{7}\hspace{-0.6mm}-\hspace{-0.6mm}z_{10}, \nonumber \\
& \hspace{-3mm}-\hspace{-0.8mm}z_{1}\hspace{-0.8mm}-\hspace{-0.8mm}z_{6}\hspace{-0.8mm}+\hspace{-0.8mm}z_{7}, \hspace{0.5mm} z_{1}\hspace{-0.8mm}-\hspace{-0.8mm}z_{6}\hspace{-0.8mm}+\hspace{-0.8mm}z_{7}, \hspace{0.5mm} \hspace{-0.8mm}-\hspace{-0.8mm}s_{1,2}\hspace{-0.8mm}+\hspace{-0.8mm}2 z_{1}\hspace{-0.8mm}+\hspace{-0.8mm}z_{3}\hspace{-0.8mm}-\hspace{-0.8mm}z_{6}\hspace{-0.8mm}+\hspace{-0.8mm}z_{7}\hspace{-0.8mm}-\hspace{-0.8mm}z_{9}\hspace{-0.8mm}-\hspace{-0.8mm}z_{10}, \nonumber \\
&\hspace{-3mm} 0, \hspace{0.5mm} 0, \hspace{0.5mm} \hspace{-0.6mm}-\hspace{-0.6mm}z_{2}\hspace{-0.6mm}-\hspace{-0.6mm}z_{6}\hspace{-0.6mm}+\hspace{-0.6mm}z_{7}, \hspace{0.5mm} 0) \,, \nonumber \\[2.5mm]
t_{6, 6} &= (0, \hspace{0.5mm} 0, \hspace{0.5mm} 0, \hspace{0.5mm} \hspace{-0.6mm}-\hspace{-0.6mm}s_{3,4}\hspace{-0.6mm}+\hspace{-0.6mm}z_{4}\hspace{-0.6mm}+\hspace{-0.6mm}z_{6}, \hspace{0.5mm} z_{5}\hspace{-0.6mm}+\hspace{-0.6mm}z_{6}, \hspace{0.5mm} 2 z_{6}, \hspace{0.5mm} \hspace{-0.6mm}-\hspace{-0.6mm}z_{1}\hspace{-0.6mm}+\hspace{-0.6mm}z_{6}\hspace{-0.6mm}+\hspace{-0.6mm}z_{7}, \nonumber \\
& \hspace{-3mm} s_{1,2}\hspace{-0.6mm}-\hspace{-0.6mm}2 z_{1}\hspace{-0.6mm}-\hspace{-0.6mm}z_{3}\hspace{-0.6mm}+\hspace{-0.6mm}z_{6}\hspace{-0.6mm}+\hspace{-0.6mm}z_{8}\hspace{-0.6mm}+\hspace{-0.6mm}z_{9}\hspace{-0.6mm}+\hspace{-0.6mm}z_{10}, \hspace{0.5mm} 0, \hspace{0.5mm} 0, \hspace{0.5mm} z_{6}\hspace{-0.6mm}+\hspace{-0.6mm}z_{11}, \hspace{0.5mm} \hspace{-0.6mm}-\hspace{-0.6mm}2) \,. \nonumber
\end{align}

\noindent In contrast, for this example the {\sc Singular} command {\tt syz}
did not produce output after 56 hours of running time with
52 GB of RAM used.

%%%%%%%%%%%%%%%%%%%%%%%%%%%%%%%%%%%%%%%%%%%%%%%%%%%%%%%%%%%%%%%%%%%%%%%%%%%%%%%%
\section{Conclusions}\label{sec:conclusions}
%%%%%%%%%%%%%%%%%%%%%%%%%%%%%%%%%%%%%%%%%%%%%%%%%%%%%%%%%%%%%%%%%%%%%%%%%%%%%%%%

Integration-by-parts (IBP) identities between loop integrals
arise from the vanishing integration of total derivatives
in dimensional regularization. The condition that
a total derivative leads to an IBP identity which does not
involve dimension shifts can be stated as the syzygy equation
\eqref{eq:syzygy_equation}. We presented in
eqs.~\eqref{eq:syzygy_generators_components}--\eqref{eq:syzygy_generators}
an explicit generating set of solutions
of the syzygy equation, valid for any number of loops
and external momenta. In general, S-polynomial computations would
be required in order to obtain the syzygies. However, as
the Baikov polynomial \eqref{eq:definition_of_Baikov_polynomial}
is the determinant of a matrix, a generating set of syzygies can be
obtained from the Laplace expansion of the determinant.
Moreover, we showed that the syzygies needed for IBP identities
evaluated on a generalized-unitarity cut can be obtained immediately from
eqs.~\eqref{eq:syzygy_generators_components}--\eqref{eq:syzygy_generators}
by a straightforward module intersection computation.

We emphasize that the closed-form expressions in
eqs.~\eqref{eq:syzygy_generators_components}--\eqref{eq:syzygy_generators}
are valid for any number of loops and external legs. The only
quantities that depend on the graph in question
are the relations of the $z$-variables to the $x$-variables
in eqs.~\eqref{eq:relation_of_z_to_x}--\eqref{eq:definition_of_z}.
In particular, the closed-form expressions allow
the construction of purely $D$-dimensional IBP identities
in cases where S-polynomial based computations of syzygies
are not feasible. An example of the latter is the non-planar
double-pentagon diagram considered in Sec.~\ref{sec:example_NP_double_pentagon}.

Moreover, an important feature of the syzygies
eqs.~\eqref{eq:syzygy_generators_components}--\eqref{eq:syzygy_generators}
is that they are guaranteed to have degree at most one. In contrast, a generating
set of syzygies obtained from an S-polynomial computation would
in general have higher degrees. The fact that the syzygies obtained
here are of degree at most one is useful because as a result
the computation of solutions which satisfy further constraints
is dramatically simplified. For example, one may be interested in
imposing the further constraint on the total derivatives that
no integrals with squared propagators
are encountered in the IBP identities.

It is worth pointing out that ref.~\cite{Kosower:2018obg}
makes use of syzygies to construct IBP identities which involve
arbitrary numerator powers. These are then solved as difference equations
to obtain the IBP reductions. Based on preliminary tests of several two-loop examples,
the syzygies from Laplacian expansion considered here
can also produce recursive relations similar to the relations
in ref.~\cite{Kosower:2018obg}. This direction merits further investigation.

%%%%%%%%%%%%%%%%%%%%%%%%%%%%%%%%%%%%%%%%%%%%%%%%%%%%%%%%%%%%%%%%%%%%%%%%%%%%%%%%
\section*{\uppercase{Acknowledgments}}
%%%%%%%%%%%%%%%%%%%%%%%%%%%%%%%%%%%%%%%%%%%%%%%%%%%%%%%%%%%%%%%%%%%%%%%%%%%%%%%%

We thank Roman N. Lee for collaboration at an early stage of this work.
We also thank Jorrit Bosma, James Drummond, Johannes Henn, Harald Ita,
David Kosower, David Mond, Costas Papadopoulos and Mao Zeng for useful discussions.
The research leading to these results has received
funding from Swiss National Science Foundation (Ambizione grant PZ00P2
161341) and funding from the European Research Council (ERC) under
the European Union's Horizon 2020 research and innovation programme (grant agreement No 725110).
The work of Y.Z.~is also partially supported by the
Swiss National Science Foundation through the NCCR SwissMap, Grant number 141869.
The work of A.G.~is supported by the Knut and Alice Wallenberg Foundation under Grant \#2015-0083.
The work of K.J.L.~is supported by ERC-2014-CoG, Grant number 648630 IQFT.
The work of J.B.~and M.S.~was supported by Project II.5 of SFB-TRR 195
``Symbolic Tools in Mathematics and their Application'' of the German Research Foundation (DFG).

%%%%%%%%%%%%%%%%%%%%%%%%%%%%%%%%%%%%%%%%%%%%%%%%%%%%%%%%%%%%%%%%%%%%%%%%%%%%%%%%
\bibliographystyle{h-physrev}
\bibliography{Syzygies_Laplace}
%%%%%%%%%%%%%%%%%%%%%%%%%%%%%%%%%%%%%%%%%%%%%%%%%%%%%%%%%%%%%%%%%%%%%%%%%%%%%%%%
\end{document}